\documentclass[12pt]{article}
\pdfoutput=1
\usepackage[utf8]{inputenc}
\usepackage[ngerman,english]{babel}
\usepackage{amsmath,amsfonts,amssymb}
\usepackage{dsfont}
\usepackage{amsthm}
\usepackage{stmaryrd}
\usepackage{lmodern}
\usepackage{fancybox}
\usepackage{a4wide}
\usepackage{hyperref}
\usepackage{enumerate}
\usepackage{color}
\usepackage{bigints}
\usepackage{relsize}
\usepackage{graphicx}
\usepackage{float}
\usepackage{etoolbox}
\usepackage{xcolor}
\usepackage{caption}
\usepackage{booktabs}
\usepackage{chngcntr}
\usepackage{todonotes}
\usepackage[nottoc]{tocbibind}
\usepackage{geometry}
\usepackage{tikzsymbols}
\usepackage{accents}
\usepackage[para]{footmisc}
\usepackage{authblk}
\usepackage{blindtext}
\usepackage{footmisc}
\usepackage{subcaption}

\DeclareFontFamily{OMX}{lmex}{}
\DeclareFontShape{OMX}{lmex}{m}{n}{<-> lmex10}{}

\counterwithin{figure}{section}
\counterwithin{equation}{section}
\counterwithin{table}{section}

\allowdisplaybreaks

\hypersetup{pdfstartview={XYZ null null 1.20}}

\DeclareMathOperator{\R}{\mathbb{R}}

\DeclareMathOperator{\E}{\mathbb{E}}
\DeclareMathOperator{\Wkeit}{P}
\DeclareMathOperator{\argmax}{argmax}
\DeclareMathOperator{\argmin}{argmin}

\DeclareMathOperator{\Varianz}{Var}
\DeclareMathOperator{\Covarianz}{Cov}

\DeclareMathOperator{\IET}{IET}

\newcommand{\Covmulti}[1]{\Covarianz\left[#1\right]}

\newtheorem{Theorem}{Theorem}[section]

\newtheorem{Proposition}[Theorem]{Proposition}

\newtheorem{Assumption}{Assumption}[section]

\theoremstyle{definition}

\newtheorem{Remark}{Remark}[section]

\usepackage[babel]{csquotes}
\usepackage[
backend=bibtex,						
style=authoryear,					
mcite=true,
maxnames=100,						
maxcitenames=3,						
uniquename=false
]{biblatex}
\DeclareNameAlias{author}{last-first}

\bibliography{Kirch_Klein_Data_Segmentation_Invariance}

\defbibheading{head}{\section*{Literaturverzeichnis}}
\DefineBibliographyStrings{ngerman}{
	andothers = {{et\,al\adddot}},            
}
\title{\textbf{Moving sum data segmentation for stochastics processes based on invariance}}

\begin{document}
\author{\Large Claudia Kirch\thanks{Institute for Mathematical Statistics, Department of Mathematics, Otto-von-Guericke University Magdeburg, Center for Behavioral Brain Sciences (CBBS); claudia.kirch$ @ $ovgu.de}, \   
	\Large Philipp Klein\thanks{Institute for Mathematical Statistics, Department of Mathematics, Otto-von-Guericke University Magdeburg, philipp.klein$ @ $ovgu.de}
}
	\DeclareFieldFormat{citehyperref}{%
		\DeclareFieldAlias{bibhyperref}{noformat}
		\bibhyperref{#1}}
	
	\DeclareFieldFormat{textcitehyperref}{%
		\DeclareFieldAlias{bibhyperref}{noformat}
		\bibhyperref{%
			#1%
			\ifbool{cbx:parens}
			{\bibcloseparen\global\boolfalse{cbx:parens}}
			{}}}

	\savebibmacro{cite}
	\savebibmacro{textcite}
	
	\renewbibmacro*{cite}{%
		\printtext[citehyperref]{%
			\restorebibmacro{cite}%
			\usebibmacro{cite}}}
	
	\renewbibmacro*{textcite}{%
		\ifboolexpr{
			( not test {\iffieldundef{prenote}} and
			test {\ifnumequal{\value{citecount}}{1}} )
			or
			( not test {\iffieldundef{postnote}} and
			test {\ifnumequal{\value{citecount}}{\value{citetotal}}} )
		}
		{\DeclareFieldAlias{textcitehyperref}{noformat}}
		{}%
		\printtext[textcitehyperref]{%
			\restorebibmacro{textcite}%
			\usebibmacro{textcite}}}
	
	\DeclareCiteCommand{\citeauthor}
	{\boolfalse{citetracker}%
		\boolfalse{pagetracker}%
		\usebibmacro{prenote}}
	{\ifciteindex
		{\indexnames{labelname}}
		{}%
		\printtext[bibhyperref]{\printnames{labelname}}}
	{\multicitedelim}
	{\usebibmacro{postnote}}

	\DeclareCiteCommand{\citeyear}
	{}
	{\bibhyperref{\printdate}}
	{\multicitedelim}
	{}

	\DeclareCiteCommand{\citetitle}
	{\bibsentence
		\usebibmacro{prenote}}
	{\printtext[bibhyperref]{\usebibmacro{title}}}
	{\multicitedelim}
	{\usebibmacro{postnote}}

	\newcommand*{\fullref}[1]{\hyperref[{#1}]{\hyperref{#1}(\ref*{#1})}}

	\newcommand{\lam}{\lambda}
	\newcommand{\intl}{\int \limits}
	\newcommand{\eps}{\vartau}
	\newcommand{\limm}{\lim \limits}   
	\newcommand{\bigcupp}{\bigcup \limits}   
	\newcommand{\bigcapp}{\bigcap \limits}   
	\newcommand{\summ}{\sum \limits}

	\newcommand{\bB}{\mathbf{B}}

	\newcommand{\bmu}{\boldsymbol{\mu}}
	\newcommand{\bSigma}{\boldsymbol{\Sigma}}
	\newcommand{\bW}{\mathbf{W}}
	\newcommand{\bR}{\mathbf{R}}
	\newcommand{\bM}{\mathbf{M}}
	\newcommand{\bZ}{\mathbf{Z}}
	\newcommand{\bu}{\mathbf{u}}
	\newcommand{\bY}{\mathbf{Y}}
	\newcommand{\bd}{\mathbf{d}}
	\newcommand{\bX}{\mathbf{X}}
	\newcommand{\bRc}{\widetilde{\mathbf{R}}}

	\newcommand{\Rstat}{{\tt{R}} }
	\newcommand{\Rstudio}{{\tt{RStudio}} }

	\newcommand{\lk}{\left\lbrace}
	\newcommand{\rk}{\right\rbrace}
	\newcommand{\geschweift}[1]{\left\lbrace #1 \right\rbrace}
	\newcommand{\ceil}[1]{\left\lceil #1 \right\rceil}
	\newcommand{\floor}[1]{\left\lfloor #1 \right\rfloor}
	\newcommand{\linksoffen}[1]{\left( #1 \right]}
	\newcommand{\rechtsoffen}[1]{\left[ #1 \right)}

	\newcommand{\Anz}[2]{\process_{\left(#1 T,#2 T\right]}}
	\newcommand{\Anzreskaliert}[2]{\process_{\left(#1,#2\right]}}
	\newcommand{\Anznormiert}[2]{\processnormalized_{\left(#1,#2\right]}}
	\newcommand{\Verteilungskonvergenz}{\overset{\mathcal{D}}{\longrightarrow}}
	\newcommand{\Stochastischekonvergenz}{\overset{P}{\longrightarrow}}

	
	\newcommand{\EW}[1]{\E\left[#1\right]}
	\newcommand{\Var}[1]{\Varianz\left[#1\right]}
	\newcommand{\Cov}[2]{\Covarianz\left[#1,#2\right]}
	\newcommand{\Prob}[1]{\Wkeit\left(#1\right)}
	\newcommand{\supp}[1]{\underset{#1}{\sup}}
	\newcommand{\limsupp}[1]{\underset{#1}{\limsup}}
	\newcommand{\inff}[1]{\underset{#1}{\inf}}
	\newcommand{\liminff}[1]{\underset{#1}{\liminf}}
	\newcommand{\maxx}[1]{\underset{#1}{\max}}
	\newcommand{\minn}[1]{\underset{#1}{\min}}
	\newcommand{\argmaxx}[1]{\underset{#1}{\argmax}}
	\newcommand{\argminn}[1]{\underset{#1}{\argmin}}
	\newcommand{\abs}[1]{\left|#1\right|}

	\newcommand{\process}{X}
	\newcommand{\processnormalized}{Y}
	\newcommand{\stddev}{\sigma}
	\newcommand{\stddevlimit}{\rho}
	\newcommand{\meanchange}{d}
	\newcommand{\mean}{\mu}
	\newcommand{\cporig}{c}
	\newcommand{\cp}{\xi}
	\newcommand{\cpdist}{\Delta}
	\newcommand{\bandwidthorig}{h}
	\newcommand{\bandwidth}{\gamma}
	\newcommand{\signal}{\mathbf{m}}

	\newcommand{\noise}{\boldsymbol{\Lambda}}
	\newcommand{\signalsquareterms}{D}
	\newcommand{\diffusionsquareterms}{\mathbf{N}}
	\newcommand{\cpestimator}{\widehat{\cp}}
	\newcommand{\cpestimatorinterval}{\widehat{\cp}}
	\newcommand{\cpestimatorlocalmax}{\widehat{\cp}}
	\newcommand{\cporigestimator}{\widehat{\cporig}}
	\newcommand{\numbercp}{q}
	\newcommand{\numbercpestimator}{\widehat{q}}
	\newcommand{\MOSUM}{G}
	\newcommand{\MOSUMrealtime}{G}
	\newcommand{\Wcirc}{\accentset{\circ}{W}}
	\newcommand{\MOSUMorder}{G}
	\newcommand{\exceedingmoment}{\delta}
	\newcommand{\localchangefctone}{f}
	\newcommand{\localchangefcttwo}{g}
	\newcommand{\stddevMOSUM}{\Sigma}
	\newcommand{\stddevMOSUMestimator}{\widehat{\Sigma}}
	\newcommand{\localmaxcriterion}{\eta}
	\newcommand{\intervalcriterion}{\eps}
	\newcommand{\thr}{\beta}
	\newcommand{\mosum}{MOSUM}
	
	\newcommand{\dummytext}{\textcolor{white}{a}}
	\newcommand{\zentrierenalt}[1]{\vcenter{\hbox{$\displaystyle #1 $}}}
	\newcommand{\zentrieren}[1]{\vcenter{\hbox{$ \overset{\dummytext}{#1} $}}}

\maketitle


\begin{center}
	\begin{minipage}{0.875\textwidth}
\begin{center}
		\textbf{Abstract}\\
\end{center}
The segmentation of data into stationary stretches also known as multiple change point problem is important for many applications in time series analysis as well as signal processing. Based on strong invariance principles, 
we analyse data segmentation methodology using moving sum (\mosum) statistics for a class of regime-switching multivariate processes where each switch results in a change in the drift. In particular, this framework includes the data segmentation of multivariate partial sum, integrated diffusion and renewal processes even if the distance between change points is sublinear. 
We study the asymptotic behaviour of the corresponding change point estimators, show consistency and derive the corresponding localisation rates which are minimax optimal in a variety of situations including an unbounded number of changes in Wiener processes with drift. Furthermore, we derive the limit distribution of the change point estimators for local changes --  a result that can in principle be used to derive confidence intervals for the change points.  \end{minipage}
\end{center}
\noindent\textbf{Keywords:} Data segmentation, Change point analysis, moving sum statistics, multivariate processes, invariance principle, regime-switching processes\\\\
\textbf{MSC2020 classification:} 62M99; 62G20, 62H12
	
\section{Introduction}\label{introduction}
The detection and localisation of structural breaks has a long tradition in statistics, dating back to \cite{Page}. Nevertheless, there is still a large maybe even increasing interest in this topic surely also because 
change point analysis  is broadly applicable in a number of fields such as neurophysiology (see \cite{Messer14}), genomics (compare \cite{Olshen}, \cite{NiuZhang}, \cite{LiMunkSieling}, \cite{ChanChen}), finance (\cite{Aggarwal}, \cite{ChoFryzlewicz}), astrophysics (see \cite{FischEckleyFearnhead}) or oceanographics (\cite{Killick10}). 

A large amount of research deals with the detection of changes in univariate time series in particular changes in the mean (compare \cite{CsorgoHorvath97} for an overview) where recently also applications to continuous time stochastic processes,  functional or high-dimensional panel data (see e.g.\ \cite{horvath2014extensions}). However, extensions to the multiple change point problem that aims at segmenting the data into stationary stretches beyond changes in the mean in time series data are much more scarce.

Generally, data segmentation methods can roughly be split up in two approaches: The first approach first introduced by \cite{Yao} in the context of i.i.d.\ normally distributed data using the Schwarz' criterion aims at optimizing suitable objective functions.
\cite{Kuehn} extended this approach to processes in a setting closely related to the one in this paper albeit only allowing for univariate processes and a finite number of change points. Further approaches include e.~g.\ least-squares (\cite{YaoAu}) or  the quasi-likelihood-function (\cite{Braun}). Generally, such approaches are computationally expensive, such that there is another body of work proposing fast algorithms e.g.\ using dynamic programming (\cite{Killick12}, \cite{Maidstone}).

A second approach is based on hypothesis testing, where e.g.\ binary segmentation introduced by \cite{Vostrikova} recursively uses tests constructed for the at-most-one-change situation. This arises several problems including the observation that detection power can be poor if the set of change points is unfavourable, such that several extensions have been proposed in the literature such as circular binary segmentation (\cite{Olshen}) or wild binary segmentation (\cite{Fryzlewicz}).

\subsection*{Connection to existing work}
Another class of test-based methods uses moving sum (\mosum) statistics which were first introduced by \cite{BauerHackl}. They are particularly useful in the context of localising multiple change points and recently have broadly been used for the detection and  the estimation of change points, see e.g.\ \cite{yau2016inference} for changes in autoregressive time series,
\cite{EichingerKirch} in a hidden Markov framework and \cite{ChoKirchMeier} as well as
\cite{ChoKirch} who proposed a two-stage data segmentation procedure based on multiscale \mosum\  statistics. This work extends the results of \cite{EichingerKirch} to a more general setting including multivariate mean changes, changes in diffusion as well as in renewal processes. Our results also lays the foundations for the analysis of a two-step procedure as in \cite{ChoKirch}.
	 

	 \cite{Messer14} propose a bottom-up-approach combining several moving-sum statistics to obtain change point estimators in univariate renewal processes. 
	 Our work extends these results in several ways: First, \cite{Messer14} do not show consistency of the change point estimators neither do they derive localisation rates, which is one of the main results of this work. Furthermore, in addition to results for \mosum\  procedures with linear bandwidth (in the sample size) as in \cite{Messer14} we obtain results for sublinear bandwidths allowing in particular to obtain consistent estimators in situations where the distance between change points is sublinear. Additionally, we go beyond the univariate case including some multivariate point processes based on renewal processes in our analysis. 
	 Sequential change point methodology for renewal processes has been proposed by~\cite{GutSteinebach02} and \cite{GutSteinebach09}, for diffusion processes by \cite{Mihalache}.
	
	 We analyse a more general model of regime-switching multivariate processes including multivariate partial sum, renewal as well as diffusion processes.  We require the processes to fulfill a multivariate invariance principle, where processes switch (possibly with a number increasing to infinity with increasing sample size)  between finitely many regimes with each switch resulting in a change in the drift.
	 A univariate version of that model with  at-most-one change point has  been considered by \cite{HorvathSteinebach} and \cite{KuehnSteinebach}. A univariate version for finitely many change points has been considered by \cite{Kuehn} where consistency for the number of change points has been shown. Those results are now extended to include \mosum\  methodology for the estimation of a multiple (possibly unbounded) number of change points in a multivariate setting, where we achieve a minimax optimal separation rate in addition to a minimax optimal localisation rate (for the change point estimators) in case of a bounded number of change points as well as for Wiener processes with drift (see Remark~\ref{Remark_minimax} below).

	  \subsection*{Organization of the material}
	  In Subsection \ref{subsection_model}, we introduce the multiple change point model we consider followed by some examples of processes fulfilling the model in Subsection \ref{subsection_examples}. In Section \ref{section_procedures}, we describe how to  estimate change points based on \mosum\  statistics: First, we introduce the \mosum\  statistics in  \ref{subsection_mosum_statistics}, before  presenting the estimators for the structural breaks in  \ref{subsection_cp_estimators}. In \ref{subsection_threshold} we derive some asymptotic results for the \mosum\  statistics that are required for threshold selection and can also be used in a testing context. In Section \ref{section_consistency} we show that the corresponding data segmentation procedure is consistent. Finally,  we derive the localisation rates in addition to the corresponding asymptotic distribution of the change point estimators for local changes. In Section \ref{subsection_simulation}, we present some results from a small simulation study. The proofs can be found in Appendix \ref{appendix}.

\section{Multiple change point problem}\label{section_seg_procedure}
In this section we introduce the  general multiple change model for which we derive the theoretic results.  In particular, this model includes changes in multivariate renewal processes as a special case which was the original motivation for this work.
\subsection{Model} \label{subsection_model}
Consider $P<\infty$ time-continuous $p$-dimensional stochastic processes $\{\bR_{t,T}^{(j)}:0\le t\le T\}$ with (unknown) drift $(\bmu_T^{(j)}\cdot t)$ and (unknown) covariance $(\bSigma_{j,T}\cdot t)$ fulfilling the following joint invariance principle.

The observed process is then assumed to switch between these $P$ processes (states).

\begin{Assumption}\label{invariance_principle}\ \\
	Denote the joint process by $\bR_{t,T}=\left({\bR_{t,T}^{(1)}}^{\prime},\ldots,{\bR_{t,T}^{(P)}}^{\prime}\right)^{\prime}$ as well the joint drift by $\bmu_T=\left({\bmu_T^{(1)}}^{\prime},\ldots,{\bmu_T^{(P)}}^{\prime}\right)^{\prime}$, where $^{\prime}$ indicates the matrix transpose.
	 For every $ T>0 $ there exist $(p\cdot P)$-dimensional Wiener processes $\bW_{t,T}$ with covariance matrix $\bSigma_T$ and	\begin{align*}
		\bSigma_{T}^{(i)}=(\bSigma_T(l,k))_{l,k=p\,(i-1)+1,\ldots,p\,i}
	\end{align*}
	with
	\begin{align*}
		\|{\bSigma_{T}^{(i)}}\|=O(1),\qquad \|{\bSigma_{T}^{(i)}}^{-1}\|=O(1),
\end{align*}
	such that, possibly after a change of probability space,  it holds that for some sequence $\nu_T\to 0$
	\begin{align*}
	\sup_{0\le t\le T}\|\bRc_{t,T}-\bW_{t,T}\|=
	\sup_{0\le t\le T}\|\left(\bR_{t,T}-\bmu_T\,t\right)-\bW_{t,T}\|=		O_P\left(T^\frac{1}{2}\,\nu_T\right),
	\end{align*}
	where $\bRc_{t,T}=\bR_{t,T}-\bmu_T\,t$ denotes the centered process.
\end{Assumption}
If these $P$ processes  are independent, which is a reasonable assumption in a switching context, the joint invariance principle reduces to the validity of an invariance principle for each process.

The assumption on the norm of the covariance matrices is equivalent to the smallest eigenvalue of $\bSigma_T^{(i)}$ being bounded in addition to being bounded away from zero (both uniformly in $T$). In many situations, the covariance matrices will not depend on $T$, in which case this assumption is automatically fulfilled under positive definiteness. The convergence rate $\nu_T$  in the invariance principle typically depends on the number of moments that exist. Roughly speaking, the more moments the original process has, the faster $\nu_T$ converges.

We now observe a process $\bZ_{t,T}$ with increments switching between the above processes  at some unknown change points $ 0=\cporig_{0}<\cporig_{1}<\ldots<\cporig_{q_T}<\cporig_{q_T+1}=T$, where $q_T$ can be bounded or unbounded.
More precisely, we observe for $\cporig_{\ell}<t\le \cporig_{\ell+1}$
\begin{align}
	\bZ_{t,T}=\left(\bR_{t,T}^{(\cporig_{\ell}+1)}-\bR_{\cporig_{\ell},T}^{(\cporig_{\ell}+1)}\right)+\sum_{j=1}^{\ell
}\left(\bR_{\cporig_j,T}^{(\cporig_j)}-\bR_{\cporig_{j-1},T}^{(\cporig_j)}\right).
	\label{eq_model}
\end{align}
The upper index $^{(\cporig_j)}$ at the process $\bR_{\cdot,T}$ indicates (with a slight abuse of notation) the active regime between the $(j-1)$-st and the $j$-th change point, from which the increments come in that stretch.
Because we concentrate on the detection of changes in the drift, we need to assume that the drift changes between two neighboring regimes,
 i.e.\
\begin{align*}
	\bd_{i,T}:=\bmu_T^{(\cporig_i+1)}-\bmu_T^{(\cporig_i)} \neq 0 \quad \text{for all }i=1,\ldots,q_T,
\end{align*}
where $\bd_{i,T}$ is bounded but we allow for $\bd_{i,T}\to 0$ as long as the convergence is slow enough (see Assumption \ref{bandwidth}).
For ease of notation we drop the dependency on $T$ for all above quantities except $q_T$ in the following except in situations where it helps clarify the argument. The aim of this paper is to estimate the number and location of the change points and prove consistency of the estimator for the number of change points in addition to deriving localisation rates for the change point estimators.

The corresponding univariate model with at most one change was first considered by \cite{HorvathSteinebach} and extended to a gradual change setting by \cite{steinebach2000some}.  \cite{KirchSteinebach} 
prove validity for  corresponding permutation tests. Furthermore, \cite{GutSteinebach02} 
develop sequential change point tests and analyse the corresponding stopping time (\cite{GutSteinebach09}).

A related univariate multiple change situation
 with a bounded number of change points has been considered by  \cite{KuehnSteinebach}, who propose to use a Schwarz information criterion for change point estimation. However, this methodology is computationally expensive with quadratic computational complexity, which is one of the reasons why we propose an alternative methodology  based on a single-bandwidth moving sum (\mosum) statistic in order to estimate the change points. We will show that the rescaled change point estimators are consistent and derive the corresponding localisation rates.

\subsection{Examples}\label{subsection_examples}
In this section, we give three important examples fulfilling the above model assumptions, namely partial sum-processes, renewal processes as well as integrals of diffusion processes including Ornstein-Uhlenbeck and Wiener processes with drift. A detailed analysis of \mosum\  procedures  for detecting changes in (univariate) renewal processes extending the work by ~\cite{Messer14} was the original motivation for this work and is covered by this much broader framework.

\subsubsection{Partial-Sum-Processes}\label{sec_partial_sums}
This first example extends the classical multiple changes in the mean model:

Let $\bX^{(i)}_1,\bX^{(i)}_2,\ldots$ be a time series with $ \EW{\bX^{(i)}_1}=0$ and $\Covmulti{\bX^{(i)}_1}=I_p$  and all $i=1,\ldots,P$.
Let 
\begin{align*}
	\bR_{t}^{(i)}=\sum_{j=1}^{\lfloor t\rfloor}\left(\bmu^{(i)}+ {\bSigma_{T}^{(i)}}^{1/2}\bX^{(i)}_j\right).
\end{align*} 
The corresponding process fulfills Assumption \ref{invariance_principle} in a wide range of situations. For example, \cite{Einmahl} 
shows the validity in the case that  $\bX_1, \bX_2,\ldots$ with $\bX_j=\left(\bX^{(1)}_j,\ldots,\bX^{(P)}_j  \right)^{\prime}$ are i.i.d.~with $ \EW{\|\bX_1\|^{2+\delta}}<\infty $ for some $\delta>0$.
Additionally,  \cite{KuelbsPhilipp} state an invariance  principle for mixing random vectors in Theorem 4, additionally there are many corresponding univariate results under many different weak-dependency formulations.

For $\bX^{(i)}=\bX^{(1)}$ (and $\bSigma^{(i)}=\bSigma^{(1)}$)  for all $i$, then we are back to the classical multiple mean change problem that has been considered  in many papers in particular for the univariate situation, see e.g.~the recent survey papers by \cite{fearnhead2020relating} or \cite{cho2020data}.

\subsubsection{Renewal and some related point processes}\label{sec_ex_renewal}
The second example aims at finding structural breaks in the rates of renewal and some related point processes:

We consider $P$ independent sequences of $p$-dimensional point processes that are related to renewal processes in the following way: For each $i=1,\ldots,P$ we start with $\tilde{p}\ge p$ independent renewal processes $\widetilde{R}_{t,j}^{(i)}$, $j=1,\ldots,\tilde{p}$, from which we derive a $p$-dimensional point process $\bR_t^{(i)}=\mathbf{B}^{(i)} \left(\widetilde R_{t,1}^{(i)},\ldots,\widetilde R_{t,\tilde p}^{(i)}\right)^{\prime}$, where $\mathbf{B}^{(i)}$ is a $(p\times\tilde{p})$ - matrix with non-negative integer-valued entries.
By  Lemma 4.2 in \cite{SteinebachEastwood}  Assumption \ref{invariance_principle} is fulfilled for a block-diagonal $\bSigma_T$ with
\begin{align*}
	&\bSigma_T^{(i)}=\mathbf{B}^{(i)}\mathbf{D}\left(\frac{\boldsymbol{\sigma}^2(i)}{\boldsymbol{\mu}^{3}(i)}\right){\mathbf{B}^{(i)}}^{\prime},\\
	&\text{with }\mathbf{D}\left(\frac{\boldsymbol{\sigma}^2(i)}{\boldsymbol{\mu}^{3}(i)}\right)=\operatorname{diag}\left(\frac{\sigma_1^2(i)}{\mu_1^{3}(i)},\ldots,\frac{\sigma_{\tilde{p}}^2(i)}{\mu_{\tilde{p}}^{3}(i)}\right),
\end{align*}
 where $\mu_j(i)$ and $\sigma_j^2(i)$ are the mean and variance of the corresponding inter-event times. \cite{SteinebachEastwood} and \cite{Csenki} consider $\tilde{p}=p$ but use inter-event times that are dependent for $j=1,\ldots,p$. In such a situation, the invariance principle in  Assumption \ref{invariance_principle} still holds if the intensities are the same across components with  $\Sigma_T^{(i)}=\Sigma_{\IET}^{(i)}/\mu_1^3(i)$, where $\Sigma_{\IET}^{(i)}$ is the covariance of the vector of inter-event times -- a setting that we adopt in the simulation study. If the intensities differ, then by \cite{SteinebachEastwood} an invariance principle towards a Gaussian process can still be obtained, but this is no longer a multivariate Wiener process. While each component is a Wiener process, the increments from one component may depend on the past of another. Many of the below results can still be derived in such a situation, however, such a model does not seem to be very realistic for most applications as the stochastic behavior of the increments of one component depends on the lagged behavior of the other components, where the lag increases with time. While a lagged dependence is realistic in many situations, in most situations one would expect this lagged-dependence to be constant across time.


		\cite{Messer14} 
		consider this model for univariate renewal processes with varying variance. They propose a multiscale procedure based on \mosum\  statistics related to those we will discuss in the next section using linear bandwidths. In \cite{Messer17}, they extend the procedure to processes with weak dependencies. They show convergence in distribution of the \mosum\   statistics to functionals of Wiener processes similar to the results that we obtain and analyze the behaviour of the signal term in \cite{Messer172}. However, they have not derived any consistency results for the change point estimators. In this paper, we extend their results to sublinear bandwidths and prove the consistency of the corresponding estimators as well as their localisation rates.
\subsubsection{Diffusion processes}
Clearly, switching between independent (or components of a multivariate) Brownian motion with drift is included in this framework. Additionally,
\cite{Heunis} and  \cite{Mihalache} derive invariance principles in the context of diffusion processes including Ornstein-Uhlenbeck processes among others. Let $ \left(\bX_t\right)_{t\ge0} $ be a stochastic process in $ \R^N $ satisfying a stochastic differential equation 
\[ 
d\bX_t=\bmu\left(\bX_t\right)dt+ \bSigma\left(\bX_t\right)d\bB_t
 \]
 with respect to an $ n $-dimensional standard Wiener process $ \left(\bB_t\right)_{t\ge0} $ and let $ \bmu,\bSigma $ be globally Lipschitz-continuous. 
 Under some conditions on $ f:\R^N\to\R^p $, as given by \cite{Heunis},   relating to $ \bmu,\bSigma $, which in particular guarantee that the function $f$ applied to the (invariant) diffusion results in a centered process, 
 there exists a $ p $-dimensional Wiener process $ \left(\bW_t\right)_{t\ge0} $ and some $ \eta>0 $ such that
 \[ 
 \left\|\int_0^T f(\bX_s)\ ds-\bW_T\right\|=O\left(T^{1/2-\eta}\right),
  \]
  where $ \left(\bX_t\right)_{t\ge0} $ either is a solution to the SDE with fixed starting value $ \bX_0=y_0 $ or a strictly stationary solution  with respect to an invariant distribution.\\
  Furthermore, in the case of a one-dimensional stochastic diffusion process, \cite{Mihalache} showed for some $ L^2 $-functions fulfilling constraints depending on $ \bmu, \bSigma $ that there exists a strong invariance principle for the integrals of diffusion processes with a rate of $ O(\left(T\log_2 T\right)^{1/4}\sqrt{\log T}) $.

\section{Data segmentation procedure} \label{section_procedures}
Now, we are ready to introduce a \mosum-based data segmentation procedure for stochastic processes following the above model:

\subsection{Moving sum statistics} \label{subsection_mosum_statistics}
At every change point, the drift of the process that is active to the left of that change point is different from the drift of the  process that is active to the right of the change point. 
Consequently, the increment of the process $\bR$ to the left will systematically differ from the one to the right. On the other hand, in a stationary stretch away from any change point both increments will be approximately the same as they are estimating the same drift. It is this observation that gives rise to the following moving sum (\mosum) statistic that is based on the moving difference of increments with bandwidth $h=h_T$
\begin{align}\label{eq_mosum}
	\bM_t=\bM_{t,T,h_T}(\bZ)&=\frac{1}{\sqrt{2h}}\,\left[\left(\bZ_{t+h}-\bZ_t\right)-\left( \bZ_t-\bZ_{t-h} \right)\right]\notag\\
	&=\frac{1}{\sqrt{2h}}\,\left(\bZ_{t+h}-2\bZ_t+\bZ_{t-h}\right).
\end{align}
If there is no change, then this difference will fluctuate around 0.  On the other hand close to a change point, this difference will be different from 0. 
Ideally, the bandwidth should be chosen to be as large as possible (to get a better estimate obtained from a larger 'effective sample size' of the order $h$). On the other hand, the increments should not be contaminated by a second change as this can lead to situations where the change point can no longer be reliably localised by the signal.
Indeed, we need the following assumptions on the bandwidth for a change to be detectable:
\begin{Assumption}\ \label{bandwidth}
For $\nu_T$ as in Assumption~\eqref{invariance_principle}	the bandwidth $h<T/2$ fulfills
		\begin{align*}
&	\frac{\nu_T^2\,T\,\log T}{h}\to 0.
\end{align*}	
Furthermore, it isolates the $i$-th change point in the sense of
\begin{align}
		&		h\le \frac 1 2 \, \cpdist_i,\qquad \text{where }\cpdist_i=\min(\cporig_{i+1}-\cporig_i,\cporig_i-\cporig_{i-1}). \label{condition_bandwidth}
\end{align}	
Additionally, the signal needs to be large enough to be detectable by this bandwidth, 
i.e.\
	\begin{align}
		\frac{
	\|\bd_i\|^2\,h}{\log\left( \frac{T}{h}
	\right)}\to\infty.\label{condition_mean_change}
		\end{align}
\end{Assumption}
Combining \eqref{condition_bandwidth} and \eqref{condition_mean_change} shows that -- with an appropriate bandwidth $h$ -- changes are detectable as soon as 
\begin{align}\label{eq_separation}
		\frac{
	\|\bd_i\|^2\,\cpdist_i}{\log\left( \frac{T}{\cpdist_i}\right)
	}\to\infty.
		\end{align}
	 In case of the classical mean change model as in Subsection~\ref{sec_partial_sums} this is known to be the minimax-optimal separation rate that cannot be improved (see Proposition 1 of \cite{AriasCastro}).

	 The assumption on the distance of the first and last change point to the boundary of the process in \eqref{condition_bandwidth} can be relaxed as no boundary effects can occur there.

		\subsection{Change point estimators} \label{subsection_cp_estimators}
		The \mosum\  statistic $\bM_t= \signal_{t}+ \noise_t$ as in \eqref{eq_mosum}  decomposes into a piecewise linear signal term $ {\signal_t=\signal_{t,h,T}} $ and a centered noise term $\noise_t=\noise_{t,h,T}$ with 
		\begin{align}
		&{\sqrt{2h}}	\,\signal_t=\begin{cases}
				(h-t+c_i)\,\bd_i,&\ \text{for }c_i<t\le c_i+h,\\
				0,&\ \text{for }c_i+h<t\le c_{i+1}-h,\\
				(h+t-c_{i+1})\,\bd_{i+1},&\ \text{for }c_{i+1}-h<t\le c_{i+1},
			\end{cases} \label{signal}\\[3mm]
			&	\sqrt{2h}\;\noise_t= \sqrt{2h}\;\noise_t \label{noise}(\bRc) 	\\*
		&\quad =\begin{cases}
				\bRc_{t+h}^{(c_{i+1})}-2\bRc_{t}^{(c_{i+1})}+\bRc_{c_i}^{(c_{i+1})}
				-\bRc_{c_i}^{(c_{i})}+\bRc_{t-h}^{(c_{i})},&\ \text{for }c_i<t\le c_i+h,\\
				\bRc_{t+h}^{(c_{i+1})}-2\bRc_{t}^{(c_{i+1})}+\bRc_{t-h}^{(c_{i+1})},
				&\ \text{for }c_i+h<t\le c_{i+1}-h,\\
				\bRc_{t+h}^{(c_{i+2})}-\bRc_{c_{i+1}}^{(c_{i+2})}+\bRc_{c_{i+1}}^{(c_{i+1})}-2\bRc_{t}^{(c_{i+1})}+\bRc_{t-h}^{(c_{i+1})},
				&\ \text{for }c_{i+1}-h<t\le c_{i+1},	
			\end{cases} \notag
\end{align}
where $ \bRc_t:=\bR_t-t\bmu $ for $ i=0,\ldots,q_T $  and the upper index $c_{j}$ denotes the active regime between the $(j-1)$st and $j$-th change point (with a slight abuse of notation).\\
The
signal term is a
piecewise linear function that takes its extrema at the change points and is 0 outside $h$-intervals around the change points. Additionally,  the noise term is asymptotically negligible compared to the signal term (see Theorem \ref{theorem_threshold} for the corresponding theoretic statement and Figure \ref{figure_signal_noise} for an illustrative example).
\begin{figure}
	\includegraphics*[scale=0.44]{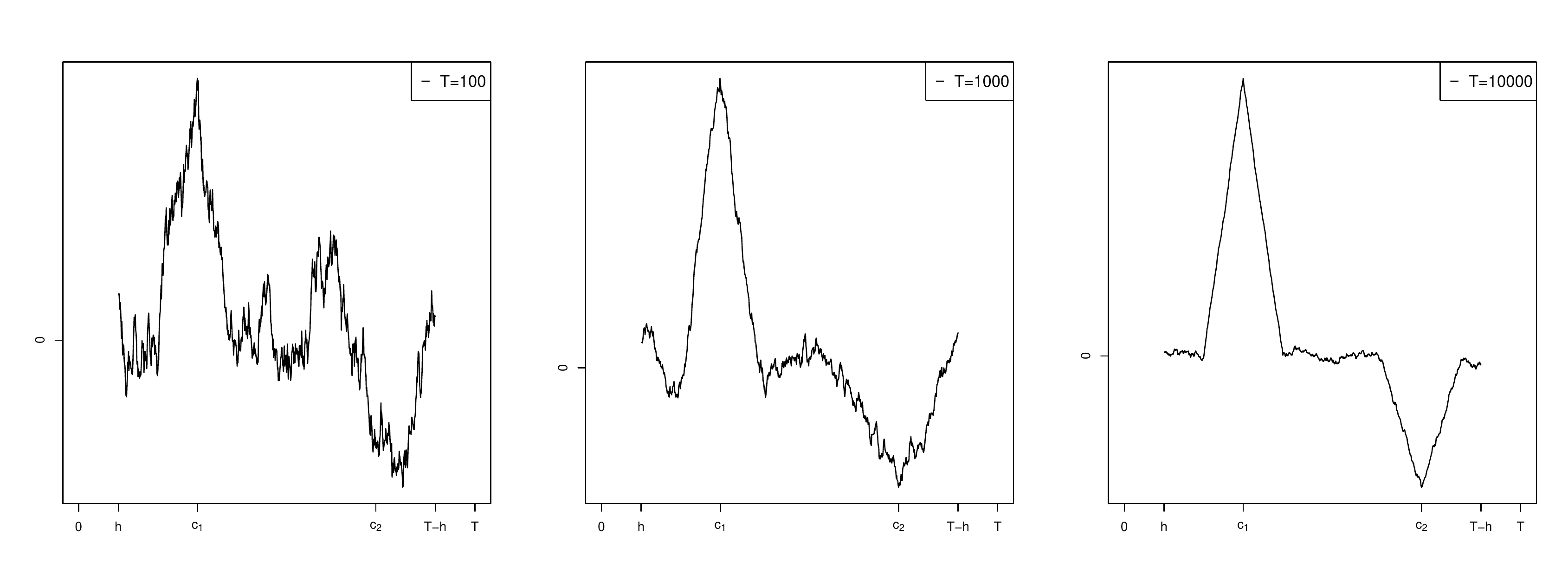}
	\caption{Univariate \mosum\  statistic with $T=100,\, 1\,000,\, 10\,000$ (from left to right), where the noise term (fluctuating around the signal)  becomes smaller and smaller relative to the signal term.}\label{figure_signal_noise}
\end{figure} 
\begin{figure}[t]
			\caption{In the upper panel, the observed event times of a univariate renewal process with 3 change points (i.e.\ 4 stationary segments) are displayed (where the plot needs to be read like a text: It starts in the upper row on the left, then continues in the first row and jumps to the second row and so on). The grey and white regions mark the estimated segmentation of the data while the red intervals mark the true segmentation.\\
		In the lower panel, the corresponding \mosum\  statistic with (relative) bandwidth $h/T=0.07$ is displayed. The grey areas are the regions where the threshold ($\alpha=0.05$ as in Remark~\ref{rem_choice_threshold})  is exceeded (in absolute value). The blue solid lines  indicate the change point estimates obtained as local extrema that fall within the grey area (making them \emph{significant}). The true change points are indicated by the red dashed lines. }\label{figure_rp_example}
	\includegraphics*[width=\textwidth]{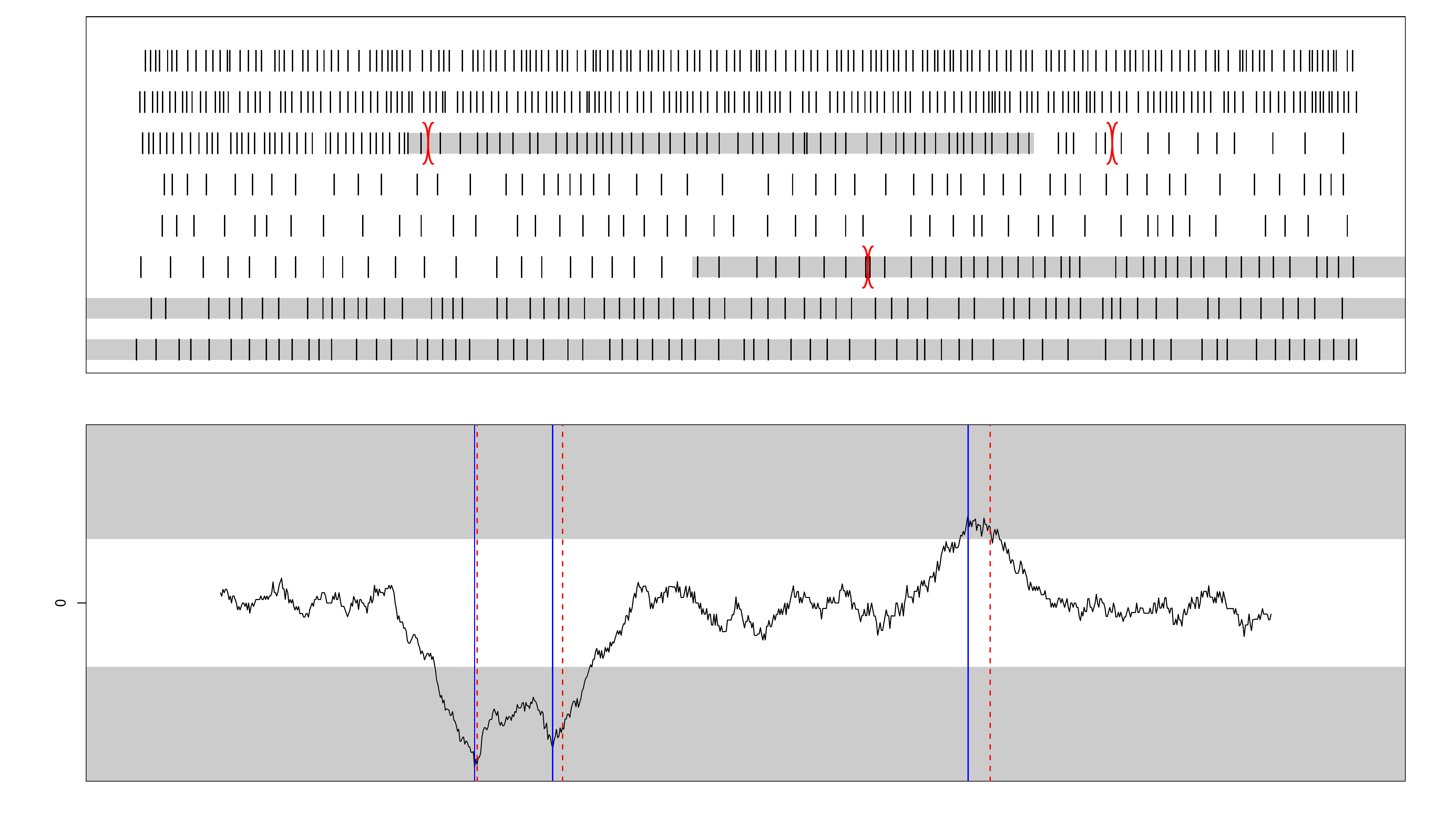}

	\end{figure}

These observations motivate the following data segmentation procedure, that considers local extrema that are big enough (in absolute value) as change point estimators: 

To this end,  a suitable threshold $\thr=\thr_{h,T}$ is needed to define \emph{significant} time points, where a point $t^*$ is \emph{significant} if 
\begin{align}\label{eq_mosum_signif}
	\bM_{t^*}^{\prime}\,\widehat{\mathbf{A}}^{-1}_{t^*}\,\bM_{t^*}\ge\thr
\end{align}
with $\widehat{\mathbf{A}}_{t^*}$ is some symmetric positive definite matrix that may depend on the data fulfilling
\begin{Assumption}
\label{variance_estimator}
\begin{align*}
	&\sup_{h\le t\le T-h}\left\|\widehat{\mathbf{A}}_{t,T}^{-1}\right\|=O_P\left(1\right),\quad 
	\sup_{i=1,\ldots,q_T}\sup_{|t-c_i|\le h}\left\|\widehat{\mathbf{A}}_{t,T}\right\|=O_P(1).
\end{align*}
\end{Assumption}
A good (non data-driven) choice  fulfilling this assumption is given by
\begin{align}\label{eq_sigma_t}
	\bSigma_{t}=\bSigma_{t,T}=\bSigma_T^{(c_{i})}
\end{align}
for $\cporig_{i-1}<t\le \cporig_{i}$, which guarantees scale-invariance of the procedure and allows for nicely interpretable thresholds (see Section~\ref{subsection_threshold}). The latter remains true for estimators as long as they fulfill
\begin{align}
	\sup_{i=1,\ldots,q_T}\sup_{|t-c_i|> h}\left\|\widehat{\bSigma}_{t,T}^{-1/2}-\bSigma_t^{-1/2}\right\|=o_P\left(\left(\log \frac{T}{h}\right)^{-1}\right) \label{variance_estimator2}
\end{align}
in addition to the above boundedness assumptions. In particular, this permits local estimators that are consistent only away from change points but contaminated by the change in a local environment thereof. The latter is typically the case for covariance estimators, think e.g.\ of the sample variance contaminated by a change point. In order to not reduce detection power in small samples, it is beneficial if the estimator is additionally consistent directly at the change point, which is also achievable (see e.g.\ \cite{EichingerKirch}).

The choice of the threshold $\thr$ will be discussed in the next section, where we can make use of asymptotic considerations in the no-change situation. 

Typically, there are intervals of significant points (due to the continuity of the signal) such that only local extrema of such intervals actually indicate a change point. To define what a local extremum is, we require  a tuning parameter $0<\localmaxcriterion <1 $. This parameter defines the locality requirement on the extremum, where a  point $t^*$ is a local extremum if it maximizes the absolute  \mosum\  statistic within its $\localmaxcriterion h$-environment, i.e.\ if
\begin{align}
	t^*=\min\left\{ \argmaxx{t^*-\eta h\le t\le t^*+\eta h}\|\bM_t\|.
	\right\} \label{cp_estimator}
\end{align}

The threshold $\beta$  distinguishes between \emph{significant}  and \emph{spurious} local extrema that are purely associated with the noise term.
The set of all significant local extrema is the set of change point estimators with its cardinality an estimator for the number of the change points.

	Figure~\ref{figure_rp_example} shows an example illustrating these ideas: Away from the change points the \mosum\  statistic fluctuates  around 0 (within the white area that is beneath the threshold in absolute value) while it falls within the grey area close to the change points -- making corresponding local extrema significant. 
	Furthermore, the statistic does not need to return to the white area in order to have all changes estimated, as can be seen between the first and second change point. This is one of the major advantages of the $\eta$-criterion based on \emph{significant} local maxima as described here (in comparison to the $\epsilon$-criterion originally investigated by \cite{EichingerKirch} in the context of mean changes, see also the discussion in \cite{ChoKirchMeier}). Nevertheless, results for the $\epsilon$-criterion can be obtained along the lines of our proofs below.

 	\subsection{Threshold selection}\label{subsection_threshold}
	As pointed out above we need to choose a threshold $\thr=\thr_{h,T}$ that can distinguish between significant and spurious local extrema. The following theorem gives the magnitudes of signal as well as noise terms:
	\begin{Theorem}\label{theorem_threshold}\
		Let the Assumptions~\ref{invariance_principle}, \ref{bandwidth} and \ref{variance_estimator} hold.
	
	\begin{enumerate}[(a)]
		\item 	For the signal $\signal_t$ with $c_i-h<t<c_i+h$, it holds
						\begin{align*}
							\signal_t^{\prime}\widehat{\mathbf{A}}_t^{-1}\signal_t \ge \frac{1}{2\|\widehat{\mathbf{A}}_t\|}\frac{(h-\abs{t-c_i})^2}{h}\|\bd_i\|^2.
			\end{align*}
			At other time points the noise term is equal to zero.
			\item For the noise term it holds for $q_T=0$, i.e.\ in the no-change situation			\begin{enumerate}[(i)]
				\item for a linear bandwidth $h=\gamma T$ with $0<\gamma<1/2$ 			\begin{align*}
						&\sup_{\gamma T \le t\le T-\gamma t} 
	\noise_{t}^{\prime}\bSigma_{T}^{-1}\noise_{t}\\
	&\qquad\Verteilungskonvergenz \sup_{\gamma\le s\le 1-\gamma}\frac{1}{2\gamma}\left({\bB}_{s+\gamma}-2{\bB}_s+{\bB}_{s-\gamma}\right)^{\prime}\left({\bB}_{s+\gamma}-2{\bB}_s+{\bB}_{s-\gamma}\right), 
			\end{align*}
		where ${\bB} $ denotes a multivariate standard Wiener process.\\ 
			In particular, the squared noise term is of order $O_P(1)$ in this case.
	\item for a sublinear bandwidth $h/T\to 0$ but Assumption \ref{bandwidth} fulfilled, it holds that
	\begin{align*}
		a\left(\frac{T}{h}\right)\supp{h\le t\le T-h}\sqrt{\noise_t^{\prime}\bSigma_T^{-1}\noise_t }-b\left(\frac{T}{h}\right)\Verteilungskonvergenz E,
	\end{align*}
	where $ E $ follows a Gumbel distribution with $ \Prob{E\le x}=e^{-2e^{-x}} $ and 
	\begin{align*}
		a(x)=&\ \sqrt{2\log x}\\
		b(x)=&\ 2\log x+\frac{p}{2}\log\log x+\log\frac{3}{2}-\log\Gamma\left(\frac{p}{2}\right).
	\end{align*}
		In particular, the above squared noise term is of order $O_P\left( \log(T/h)  \right)$ in this case.
\end{enumerate}
The assertions remain true if an estimator for the covariance is used fulfilling \eqref{variance_estimator2} uniformly over all $h\le t\le T-h$.
\item In the situation of multiple change points, it holds that
\begin{align*}
	\sup_{h \le t\le T-h} \|\noise_t\|=&\ O_P(\sqrt{\log(T/h)}).
\end{align*}
\end{enumerate}
\end{Theorem}

 In order to obtain consistent estimators, on the one hand,
 the threshold needs to be asymptotically negligible compared to the squared signal term as in Theorem~\ref{theorem_threshold} (a). This guarantees that every change is detected with asymptotic probability 1. On the other hand, the threshold needs to grow faster than the squared noise term in Theorem~\ref{theorem_threshold} (c) so that  false positives occur with asymptotic probability 0. Hence, both conditions are fulfilled under the following assumption:

\begin{Assumption}\label{threshold} The threshold fulfills:
\begin{align*}
	\frac{\thr_{h,T}}{h_T\ \min\limits_{i=1,\ldots,q_T}\|\bd_i\|^2}&\ \to 0,\qquad
	\frac{\log\frac{T}{h_T}}{\thr_{h,T}}\ \to 0 \qquad (T\to \infty).
\end{align*}
\end{Assumption}

The following remark introduces a threshold that has a nice interpretation in connection with change point testing:
\begin{Remark}\label{rem_choice_threshold}
	A common choice for the threshold is obtained as the asymptotic $\alpha_T$-quantile obtained from Theorem~\ref{theorem_threshold} (b) for some sequence $\alpha_T\to 0$. In the simulation study in Section~\ref{subsection_simulation} we use this  threshold with $\alpha_T=0.05$.
Using the asymptotic $\alpha$-quantile (from the no-change situation) guarantees that all change point estimators are significant at a global level $\alpha$, where significant is meant in the usual testing sense. This gives this choice a nice interpretability.  In fact, Theorem~\ref{theorem_threshold}
shows that such a threshold with a constant sequence $\alpha$ yields an asymptotic test at level $\alpha$ which has asymptotic power one by Theorem~\ref{theorem_consistency}. Nevertheless, often, tests designed for the at-most-one-change as in \cite{HuskovaSteinebach00}, \cite{HuskovaSteinebach02} have a better power, but are not as good at localising change points (see Figure~1 in \cite{cho2020data} for an illustration).
\end{Remark}

\section{Consistency of the segmentation procedure}\label{section_consistency}
In this section, we will show consistency of the above segmentation procedure for both the estimators of the number and locations of the change points. 
Furthermore, we derive localisation rates for the estimators of the locations of the change points for some special cases showing that they cannot be improved in general. This is complemented by the observations that these localisation rates are indeed minimax-optimal if the number of change points is bounded in addition to observing Wiener processes with drift. Otherwise the generic rates that are obtained based solely on the invariance principle will not be tight in the sense that the proposed procedure can provide better rates than suggested by the invariance principle.

The following theorem shows that the change point estimators defined in \eqref{cp_estimator} are consistent for the number and locations of the change points.
\begin{Theorem}\label{theorem_consistency}
Let Assumptions \ref{invariance_principle}, \ref{bandwidth} -- \ref{threshold} hold.  Let $ 0<\hat{c}_1<\ldots<\hat{c}_{\hat{q}_T} $ be the change point estimators of type \eqref{cp_estimator}. Then for any $\tau>0$ it holds
\begin{align*}
	\lim_{T\to\infty}\Prob{\max_{i=1,\ldots,\min(\hat{q}_T,q_T)}\left|\hat{c}_i-\cporig_i\right|\le \tau h,
	\hat{q}_T=q_T}=\ 1.  
\end{align*}
\end{Theorem}

The theorem shows in particular that the number of change points is estimated consistently. For the linear bandwidth we additionally get consistency of the change point locations in rescaled time, while for the sublinear bandwidths we already get a convergence rate of $h/T$ towards the rescaled change points.


Under the following stronger assumptions, the localisation rates can be further improved:

\begin{Assumption}\label{Hajek_Renyi}\begin{enumerate}[(a)]
			\item It holds  for any of the centered processes $\bRc^{(j)}$ as in \eqref{noise} and any value $\theta_i=\theta_{i,T}$ (which will be $c_i$ or $c_i\pm h$ when the assumption is applied) for any sequence $D_T\ge 1$ (bounded or unbounded)
	\begin{align*}
		\sup_{\frac{D_T}{\|\bd_i\|^2}\le s \le h} \frac{\sqrt{D_T} \left\|\bRc_{\theta_i}^{(j)}-\bRc_{\theta_i\pm s}^{(j)}\right\|}{s\left\|\bd_i\right\|}=O_P(\omega_T). 
	\end{align*}
\item Let now the upper index $\theta_i$ denote the active stretch in the stationary segment $(\theta_i,\theta_i+s)$ respectively $(\theta_i-s,\theta_i)$.   Then, it holds for any sequence $D_T>0$
	\begin{align*}
		\max_{i=1,\ldots,q_T}	\sup_{\frac{D_T}{\|\bd_i\|^2}\le s \le h} \frac{\sqrt{D_T} \left\|\bRc_{\theta_i}^{(\theta_i)}-\bRc_{\theta_i\pm s}^{(\theta_i)}\right\|}{s\left\|\bd_i\right\|}=O_P(\tilde{\omega}_T). 
	\end{align*}
	\end{enumerate}
\end{Assumption}
The localisation rates of the \mosum\  procedure are determined by the  rates $\omega_n,\tilde{\omega}_n$ which need to be derived for each example separately (at least for the tight ones). In the context of partial sum processes these results are well known. For example, the suprema in (a) are stochastically bounded by the H\'aj\'ek-R\'enyi inequality which has been shown for partial sum processes even with weakly dependent errors. In that context, the assertion in (b) is fulfilled with a polynomial rate in $q_T$ (see \cite{ChoKirch}, Proposition 2.1 (c)(ii)).

\begin{Remark}\label{rem_rates}
	\begin{enumerate}[(a)]
		\item  For Wiener processes with drift $\omega_T=1$ and $\tilde{\omega}_T=\sqrt{\log(q_T)}$ (see Proposition~\ref{prop_WP_rates} below).	
		\item 	By the invariance principle in Assumption~\ref{invariance_principle}, all rates are clearly dominated by $T^{1/2}\nu_T$.			However, this is often far too liberal a bound (see Proposition 2.1 in \cite{ChoKirch} for some tight bounds in case of partial sum processes). 
		\item Often, for each regime there exist forward and backwards invariance principles from some arbitrary starting value $\theta_i$. This is the case for partial sum processes and for (backward and forward) Markov processes due to the Markov property. For renewal processes, this can be shown along the lines of the original proof for the invariance principle (\cite{CsorgoHorvathSteinebach}) because the time to the next (previous) event is asymptotically negligible; see also Example 1.2 in \cite{KuehnSteinebach}).
			In this case, the H\'aj\'ek-R\'enyi results for Wiener processes carry over (see Proposition~\ref{prop_WP_rates}) to the different processes underlying each regime, resulting in $\omega_T=1$. For the situation with a bounded number of change points this carries over to $\tilde{\omega}_T$.
	\end{enumerate}
\end{Remark}

\begin{Theorem}\label{theorem_convergence_rate}\ \\
	Let Assumptions \ref{invariance_principle}, \ref{bandwidth} -- \ref{threshold} in addition to \ref{Hajek_Renyi} hold. 
	For $\hat{q}_T<q_T$ define $\hat{c}_i=T$ for $i=\hat{q}_T+1,\ldots,q_T$.
\begin{enumerate}[(a)]
	\item
	For a single change point estimator the following localisation rate holds 		\[ 
		\left\|\bd_i\right\|^2\left|\hat{c}_{i}-c_i\right|	=O_P\left(\omega_T^2\right).
	\]
\item The following uniform rate holds true:
	\[ 
		\max_{i=1,\ldots,q_T}\left\|\bd_i\right\|^2\left|\hat{c}_{i}-c_i\right|	=O_P\left(\tilde{\omega}_T^2\right).
	\]
\end{enumerate}

\end{Theorem}

\begin{Remark}[Minimax optimality]\label{Remark_minimax}
	We have already mentioned beneath \eqref{eq_separation} that the separation rate given there is minimax optimal (see Proposition 1 of \cite{AriasCastro}).

	Minimax optimal localisation rates (derived in the context of changes in the mean of univariate time series, which is  covered by the partial sum processes in our framework) are known for a few special cases: First, the minimax optimal localisation rate for a single change point and in extension also for a bounded number of change points is given by $\omega_T=1$ in the above notation (see e.g.\ Lemma 2  in \cite{wang2020univariate}).  In particular this shows that our procedures achieves the minimax optimality in case of a bounded number of change points under weak assumptions (as pointed out in Remark~\ref{rem_rates} (c)). 
Secondly, the optimal localisation rate for unbounded change points under sub-Gaussianity (attained for partial sum process of i.i.d.\ errors) is given by $\tilde{\omega_T}=\sqrt{\log T}$ (see Proposition 6 in \cite{Verzelenetal} and Proposition 2.3 in \cite{ChoKirch}). Indeed, we match this rate for Wiener processes with drift.
\end{Remark}

The following theorem derives the limit distribution of the change point estimators for local changes which shows in particular that the rates are tight. In principle, this result can be used to obtain asymptotically valid confidence intervals for the change point locations. In case of fixed changes, the limit distribution depends on the underlying distribution of the original process (see \cite{AntochHuskova_fixed}
 for the case of partial sum processes), where the proof can be done along the same lines.
We need the following assumption:

\begin{Assumption}\label{ass_cp_distr}
	Let $\bd_i=\bd_{i,T}=\|\bd_i\| \mathbf{u}_i + o(\|\bd_i\|)$ with $\|\mathbf{u}_i\|=1$ and $\|\bd_{i,T}\|\to 0$.
Assume that	 $\bY^{(j)}_s=\bY^{(j)}_s(c_i,D)$ with
	\begin{align*}
		&		\bY^{(1)}_s=\bRc^{(c_i)}_{c_i-h+\frac{s-D}{\|d_i\|^2}}-\bRc^{(c_i)}_{c_i-h-\frac{D}{\|d_i\|^2}},\\
&		\bY^{(21)}_s=\bRc^{(c_i)}_{c_i+\frac{s-D}{\|d_i\|^2}}-\bRc^{(c_i)}_{c_i-\frac{D}{\|d_i\|^2}},\qquad	\bY^{(22)}_s=\bRc^{(c_{i+1})}_{c_i+\frac{s-D}{\|d_i\|^2}}-\bRc^{(c_{i+1})}_{c_i-\frac{D}{\|d_i\|^2}},\\
&		\bY^{(3)}_s=\bRc^{(c_{i+1})}_{c_i+h+\frac{s-D}{\|d_i\|^2}}-\bRc^{(c_{i+1})}_{c_i+h-\frac{D}{\|d_i\|^2}}
	\end{align*}
	fulfill the following multivariate functional central limit theorem for any constant $D>0$ in an appropriate space equipped with the supremum norm
	\begin{align*}
		\left\{\|\bd_i\|\,(\bY^{(1)}_s,\bY^{(21)}_s,\bY^{(22)}_s,\bY^{(3)}_s)^{\prime}:0\le s\le 2D\right\}\overset{w}{\longrightarrow}\left\{\widetilde{\bW}_s:0\le s\le 2D\right\},
	\end{align*}
	where $\widetilde{\bW}$ is a Wiener process with covariance matrix $\boldsymbol{\Xi}$ (not depending on $D$). For $-D\le t\le D$ denote $\bW_t=(\bW^{(1)}_t,\bW^{(21)}_t,\bW^{(22)}_t,\bW^{(3)}_t)^{\prime}=\widetilde{\bW}_{D+t}-\widetilde{\bW}_D$.
\end{Assumption}

By Assumption~\ref{bandwidth} it holds $h\|\bd_i\|^2\to \infty$, such that the distance $h-\frac{2D}{\|d_i\|^2}$ between $\bY^{(1)}$ and $\bY^{(2j)}$ (resp.\ between $\bY^{(2j)}$ and $\bY^{(3)}$) diverges to infinity. 
As such for processes with independent increments the processes $\bY^{(1)}$, $(\bY^{(21)},\bY^{(22)})^{\prime}$, $\bY^{(3)}$  are independent for $T$ large enough. 
Additionally, under weak assumptions such as mixing conditions this independence still holds asymptotically in the sense that $\bW^{(1)}$, $(\bW^{(21)},\bW^{(22)})^{\prime}$, $\bW^{(3)}$  are independent. 

Functional central limit theorems for these processes follow from invariance principles as in Assumption~\ref{invariance_principle} with $\bSigma_T\to \bSigma$  as long as such invariance principles still hold with an arbitrary (moving) starting value, which is typically the case (see also Remark~\ref{rem_rates} (c)). As such, it typically holds that $\boldsymbol{\Xi}^{(1)}=\boldsymbol{\Xi}^{(21)}=\bSigma^{(c_i)}$ and $\boldsymbol{\Xi}^{(3)}=\boldsymbol{\Xi}^{(22)}=\bSigma^{(c_{i+1})}$
where $\boldsymbol{\Xi}^{j}=\operatorname{Cov}(\bW^{(j)}_1)$ and $\bSigma^{(c_i)}$ is the covariance matrix associated with the regime between the $(i-1)$st and $i$th change point.

The following theorem gives the asymptotic distribution for the change point estimators in case of local change points.
\begin{Theorem}\label{limit_distribution}\ \\
	Let Assumptions \ref{invariance_principle}, \ref{bandwidth} -- \ref{threshold}, \ref{Hajek_Renyi} (a) with $\omega_T=1$  and \ref{ass_cp_distr} hold. For $\hat{q}_T<q_T$ define $\hat{c}_i=T$ for $i=\hat{q}_T+1,\ldots,q_T$.
	Let 	
	\begin{align*}
	\Psi_t^{(i)}:=-\abs{t}+\begin{cases}
		\mathbf{u}_i^{\prime}\bW_{t}^{(1)}-2\,\mathbf{u}_i^{\prime}\bW_{t}^{(21)}+ \mathbf{u}_i^{\prime}\bW_t^{(3)},&\ t<0\\	
		\mathbf{u}_i^{\prime} \bW_{t}^{(1)}-2\,\mathbf{u}_i^{\prime}\bW_{t}^{(22)}+\mathbf{u}_i^{\prime}\bW_t^{(3)},&\ t\ge 0.
	\end{cases}
	\end{align*}
	Then, for all $ i=1,\ldots,q_T, $ it holds that for $T\to\infty$
		\[ 
\left\|\bd_i\right\|^2\left(\hat{c}_{i}-c_i\right)\Verteilungskonvergenz \argmax\geschweift{\Psi_t^{(i)}\Big| t\in\R}
	\]

	If there is a fixed number of changes $q_T=q$ with $q$ fixed and a functional central limit theorem as in Assumption~\ref{ass_cp_distr} holds jointly for all $q$ change points, then the result also holds jointly.
\end{Theorem}

Due to the Markov property of Wiener processes, $\{\Psi_t^{(i)}:t\ge 0\}$ is independent of $\{\Psi_t^{(i)}:t\le 0\}$.

\begin{Remark}
	\begin{enumerate}[(a)]
		\item If $\bW^{(1)}$, $(\bW^{(21)},\bW^{(22)})^{\prime}$, $\bW^{(3)}$  are independent which is typically the case (see discussion beneath Assumption~\ref{ass_cp_distr}), then 	$\Psi_t^{(i)}$ simplifies to
\begin{align*}
	\Psi_t^{(i)}:=-\abs{t}+\begin{cases}
		\sqrt{\sigma^2_{(1)}+4\,\sigma^2_{(21)}+ \sigma^2_{(3)}}\,B_t,&\ t<0\\	
		\sqrt{\sigma^2_{(1)}+4\, \sigma^2_{(22)}+\sigma^2_{(3)}}\,B_t,&\ t\ge 0,
	\end{cases}
	\end{align*}
	where $B$ is a (univariate)  standard Wiener process and  $\sigma_{(j)}^2=\mathbf{u}_i^{\prime}\boldsymbol{\Xi}^{(j)}\mathbf{u}_i$. Usually (see discussion beneath Assumption~\ref{ass_cp_distr})  $\sigma_{(21)}=\sigma_{(1)}$ and $\sigma_{(22)}=\sigma_{(3)}$ further simplifying the expression.	For some examples such as partial sum processes it holds $\bSigma_t=\bSigma$ for all $t$, such that all $\sigma_{(j)}$ coincide. In this case this further simplifies to 
\begin{align*}
	\Psi_t^{(i)}:=-\abs{t}+\sqrt{6}\, \sigma_{(1)}\,B_{t}.
\end{align*}
For univariate partial sum processes this result has already been obtained in Theorem 3.3 of \cite{EichingerKirch}. However, the assumption of $\bSigma_t=\bSigma$ is typically not fulfilled for renewal processes because the covariance depends on the changing intensity of the process.
	\item If  $\bW^{(1)}$, $(\bW^{(21)},\bW^{(22)})^{\prime}$, $\bW^{(3)}$  are independent and $\bM_t$  in \eqref{cp_estimator} is replaced by $\bSigma_t^{-1/2}\bM_t$, then the Wiener processes $\bW^{(j)}$ are standard Wiener processes, such that $\Psi_t^{(i)}$ simplifies to
		\begin{align*}
				\Psi_t^{(i)}:=-\abs{t}+\sqrt{6}\, B_{t}.
		\end{align*}
		This shows that in this case the limit distribution of $\hat{c}_{i}-c_i$ does only depend on the magnitude of the change $\bd_i$ but not on its direction $\mathbf{u}_i$.
			
				Statistically, however, this is difficult to achieve as it requires a uniformly (in $t$) consistent estimator for the usually unknown covariance matrices $\bSigma_t$.

\end{enumerate}
\end{Remark}

\section{Simulation study}\label{subsection_simulation}
In this section, we illustrate the performance of our procedure for multivariate renewal processes by means of a simulation study. Related simulations in addition to a variety of  data examples for partial sum processes have been conducted by \cite{EichingerKirch,ChoKirchMeier} and for univariate renewal processes by \cite{Messer14,Messer17}.

More precisely, we analyse three-dimensional renewal processes with $T=1600$, where the increments of the
inter-event times for each component are 
$ \Gamma- $distributed with intensity changes at 250, 500, 900 and 1150, where the expected time  $\mu$ between events is given by $1.3$, $0.9$, $0.6$, $0.8$ and $1.3$.  We use a bandwidth of $ h=120$ and the parameter $\eta=0.75$. Smaller values of $\eta$ as suggested by \cite{ChoKirchMeier} for partial sum processes tend to produce duplicate change point estimators by having two or more significant local maxima for each change point if the variance is too large. For a single-bandwidth \mosum\ procedure as suggested here, this should be avoided but can be relaxed if a post-processing procedure is applied as e.g.\ by \cite{ChoKirch} for partial sum processes.

In contrast to partial sum processes, it is natural for renewal processes that the variances change with the intensity, therefore 
we consider the following three scenarios:  (i) standard deviations of constant value 0.7 (referred to as {\texttt{constvar}}), (ii)  standard deviations being $ 5/6  \mu $ (referred to as {\texttt{smallvar}}) and (iii) multivariate Poisson processes (referred to as {\texttt{Poisson}}).

We consider both the case of independence and dependence between the three components.
In the latter case, we generate for each regime $i$ an independent (in time) sequence of $ \Gamma- $distributed inter-event-times $Y_{j}=Y_{j}^{(i)}$, $j=1,2,3$,  with a correlation of $ 0.2 $ (for all pairs) as $ Y_{j}=X_{j}+X_{4}$, where $X_j\sim\Gamma(s,\lambda) $ for $j=1,2,3$ and $X_{4}\sim\Gamma(s/4,\lambda) $ for appropriate values of $s$ and $\lambda$ (resulting in the above intensities and standard deviations for each regime).

\begin{table}[t]
	\begin{subtable}[h]{\textwidth}
		\caption{{\texttt{constvar}}: Constant standard deviation of $0.7$.}
		\begin{tabular}{c||c|c|c|c|c|c}
			Change point at& 250 & 500 & 900 & 1150& spurious & duplicate \\ \hline\hline
			independent, estimator (A)  & 0.9992  & 0.9985 & 0.9253 & 0.9998 & 0.0243 & 0.0035\\ 
			independent, estimator (B)  & 0.9962  & 0.9727 & 0.6149 & 0.9998 & 0.0033 & 0.0003\\ \hline
			dependent, estimator (A) & 0.9966  & 0.9959 & 0.9052 & 1 & 0.0326 & 0.0030\\
			dependent, estimator (B) & 0.9879  & 0.9551 & 0.6245 &  0.9993 & 0.0072 & 0.0004\\
			dependent, estimator (C)	& 0.9439 & 0.8360 & 0.3534 & 0.9975 & 0.0049 & 0.0004\\
		\end{tabular}	
		\vspace{3mm}
		\label{Table_constvar}
	\end{subtable}
	
	\begin{subtable}[h]{\textwidth}
		\caption{{\texttt{smallvar}}: Standard deviation of $ 5/6 $ the expected time between events.}
		\begin{tabular}{c||c|c|c|c|c|c}
			Change point at & 250 & 500 & 900 & 1150& spurious & duplicate\\ \hline\hline
			independent, estimator (A) & 0.9790  & 1 & 0.9707 &  1 & 0.0273 & 0.0023\\ 
			independent, estimator (B) & 0.9354  & 1 & 0.9302 &  0.9999 & 0.0038 & 0.0002\\ \hline
			dependent, estimator (A) & 0.9657  & 0.9999 & 0.9527 &  0.9996 & 0.0347 & 0.0026\\
			dependent, estimator (B) & 0.9197  & 0.9982 & 0.9089 & 0.9986 & 0.0071 & 0.0013\\
			dependent, estimator (C)	&  0.7421  & 0.9882 & 0.7137 &  0.9896 & 0.0043 & 0.0003\\
		\end{tabular}
		\vspace{3mm}
		\label{Table_smallvar}
	\end{subtable}
	\begin{subtable}[h]{\textwidth}
		\caption{ {\texttt{Poisson}}-distributed inter-event times.}
		\begin{tabular}{c||c|c|c|c|c|c}
			Change point at& 250 & 500 & 900 & 1150& spurious & duplicate\\ \hline\hline
			independent, estimator (A) & 0.8913  & 0.9963 & 0.8615 & 0.9976 &  0.0390 & 0.0042 \\ 
			independent, estimator (B) & 0.7338  & 0.9844 & 0.7174 &  0.9876 & 0.0029 & 0.0009\\ \hline
			dependent, estimator (A) & 0.8654 & 0.9910 & 0.8331 & 0.9923 & 0.0457 & 0.0047\\
			dependent, estimator (B) & 0.7138 & 0.9756 & 0.6961 & 0.9749 & 0.0056 & 0.0012\\
			dependent, estimator (C)	& 0.4525 & 0.8883 & 0.4217 & 0.8963 & 0.0042 & 0.0001\\
		\end{tabular}
		\vspace{2mm}
		\label{Table_poisson}
	\end{subtable}
	\caption{Detection rates for each change point as well as the average number of spurious and duplicate estimators for different distributions of the inter-event times.}
	\label{Table_simstudy}
\end{table}

In the simulations, we use a threshold as in Remark~\ref{rem_choice_threshold} with $\alpha_T=0.05$. By Section~\ref{sec_ex_renewal} and \eqref{eq_sigma_t} it holds that $\bSigma_t=\Covmulti{(Y_1,Y_2,Y_3)^{\prime}}/\EW{Y_1}^3$ 
while we use the following choices for the matrix $\widehat{\mathbf{A}}_t$ as in \eqref{eq_mosum_signif}: (A) Diagonal matrix with locally estimated variances $\widehat{\bSigma}_t(j,j)$ on the diagonal, $j=1,2,3$,  (B) with the true variances $\bSigma_t(j,j)$ on the diagonal and (C) in case of dependent components (non-diagonal) true covariance matrix $\bSigma_t$. While only (A) is of relevance in applications, this allows us to understand the influence of estimating the variance on the procedure. For dependent data, the distinction between (B) and (C) is important for applications, because a good enough estimator (resulting in a reasonable estimator for the inverse) is often not available for the full covariance matrix as in (C) for moderately high or high dimensions, while it is much less problematic to estimate (B).
In (A) the variances at location $t$ are estimated as 

\begin{align}
	\widehat{\bSigma}_t(j,j)=\min\geschweift{\frac{ \widehat{\sigma}_{j,-}^2(t)}{\widehat{\mu}_{j,-}^3(t)},\frac{ \widehat{\sigma}_{j,+}^2(t)}{\widehat{\mu}_{j,+}^3(t)}}, \label{eq_var_est}
\end{align}
where  $\widehat{\sigma}_{j,\pm}^2(t)$ and $\widehat{\mu}_{j,\pm}(t)$ are the sample variance and sample mean respectively based on the inter-event times of the $j$th-component within the windows $(t-h,t]$ for $'-'$ respectively $(t,t+h]$ for $'+'$. The first and last inter-event times that have been censored by the window are not included. Using the minimum of the left and right local estimators takes into account that the variance can (and typically will) change with the intensity which has already been discussed by \cite{ChoKirchMeier} in the context of partial sum processes.

The results of the simulation study can be found in Table \ref{Table_simstudy}, where  we consider a change point to be detected if there was an estimator in the interval $ [c_i-h,c_i+h] $.  Additional significant local maxima in such an interval are called \emph{duplicate} change point estimators, while additional significant local maxima outside any of these intervals are called \emph{spurious}. The procedure performs well throughout all simulations with high detection rate, few spurious and very few duplicate estimators. The results improve further for smaller variance, in which case the signal-to-noise ratio is better.

When the diagonal matrix with the estimated variance is being used, the detection power is larger in all cases than when the true variance is being used. In case of the changes at location 900 this is a substantial improvement, such that the use of this local variance estimator can help boost the signal significantly. This comes at the cost of having an increased but still reasonable amount of spurious and duplicate change point estimators. 
\begin{figure}[tb]
	\includegraphics*[width=\textwidth]{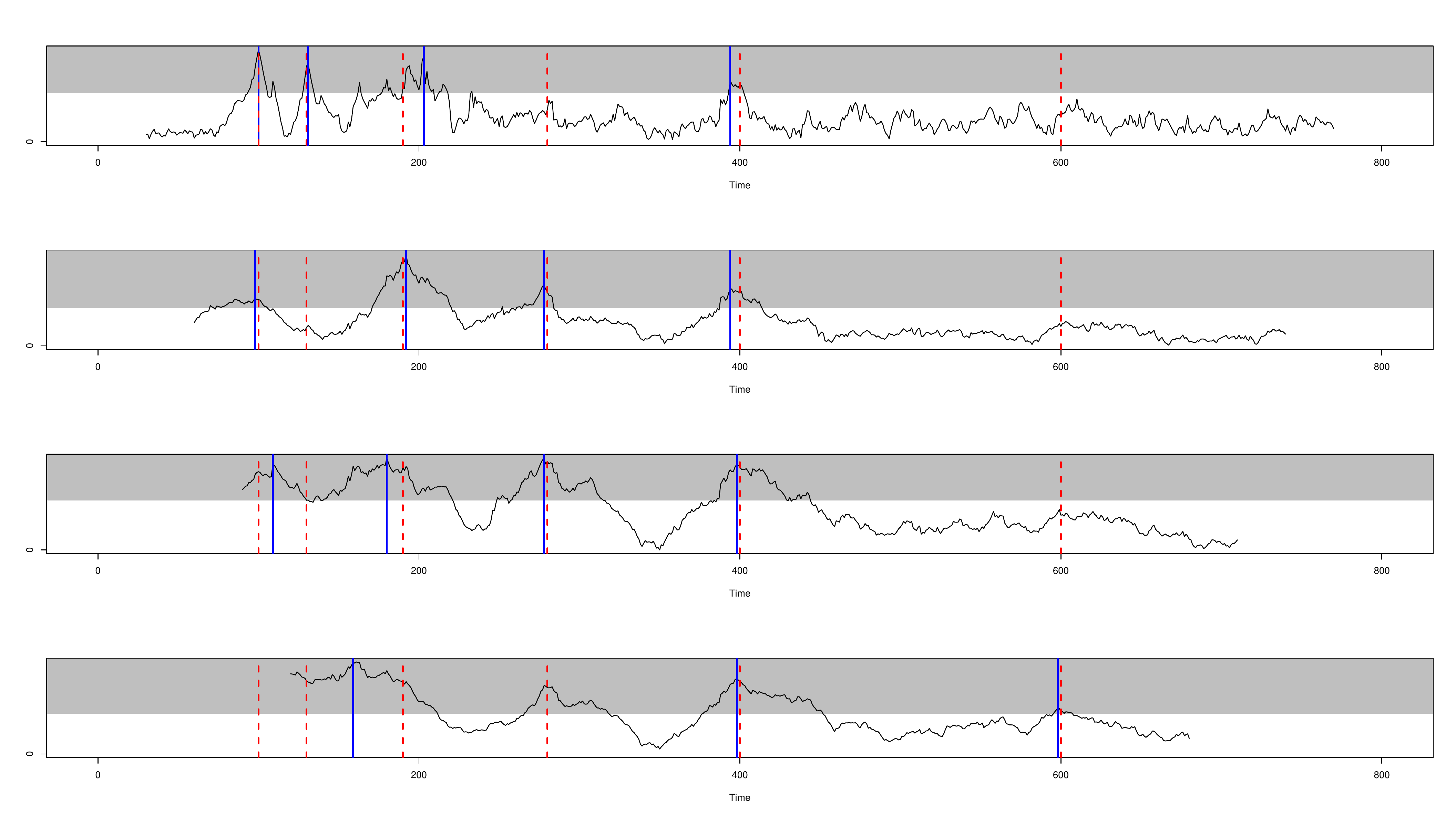}
	\caption{\mosum\ statistics with bandwidths of $ h=30,60,90,120 $ (top to bottom) for a three-dimensional renewal process with \emph{multiscale} changes
		with increasing distance between change points in combination with decreasing magnitude of the changes in intensity. The dashed vertical lines indicate the location of the true changes, while the solid lines indicate the change point estimators. \\In this multiscale situation no single bandwidth can detect all changes: The changes to the left are well estimated by smaller bandwidth, the ones in the middle by medium-sized bandwidths and the one to the right by the largest bandwidth. } \label{figure_mix}
\end{figure}
This effect stems from using the minimum in \eqref{eq_var_est}, which was introduced to gain detection power if the variance changes with the intensity. 
Additionally, the use of the true (asymptotic) covariance matrix leads to worse results than only using the corresponding diagonal matrix possibly because the asymptotic covariances do not reflect the small sample (with respect to $h$) covariances well enough. 
From a statistical perspective this is advantageous because the local estimation of the inverse of a covariance matrix in moderately large or large dimensions is a very hard problem leading to a loss in precision, while the diagonal elements are far less difficult to estimate consistently.

In the above situation the changes are \emph{homogeneous} in the sense that the smallest change in intensity is still large enough compared to the smallest distance to neighboring change points (for a detailed definition we refer to \cite{ChoKirch}, Definition 2.1, or \cite{cho2020data}, Definition 2.1).
In particular, this guarantees that all changes can be detected with a single bandwidth only.

In some applications with \emph{multiscale} signals, where frequent large changes as well as small isolated changes are present, this is no longer the case as Figure \ref{figure_mix} shows. In such cases, several bandwidths need to be used and the obtained candidates pruned down in a second step (see \cite{ChoKirch} for an information criterion based approach for partial sum processes as well as \cite{Messer14} for a bottom-up-approach for renewal processes). Similarly, if the distance to the neighboring change points is unbalanced \mosum\ procedures with asymmetric bandwidths as suggested by \cite{ChoKirchMeier} may be necessary.
\section{Discussion and Outlook}
In this paper, we considered a class of multivariate processes that, possibly after a change of probability space fulfill a uniform strong invariance principle. We assumed that the process switches possibly infinitely many times between finitely many regimes, with each switch inducing a change in the drift. This setup includes several important examples, including multivariate partial sum processes, diffusion processes as well as renewal processes.
In order to localise these changes, we extended the work of \cite{EichingerKirch} and \cite{Messer14} and proposed a single-bandwidth procedure using \mosum\  statistics in order to estimate changes, allowing for local changes. We were able to show consistency for the estimators. Further, we were able to derive (uniform) localisation rates in the form of exact convergence rates, which are indeed minimax-optimal.\\
One drawback of the procedure is the use of a single bandwidth. In practice, the identification of the optimal bandwidth turns out to be rather difficult as pointed out e.~g.~by \cite{ChoKirch} and \cite{Messer14}: On the one hand, one wants to choose a large bandwidth in order to have maximum power, while on the other hand, choosing a too large bandwidth may lead to misspecification or nonidentification of changes. Furthermore, as can be seen in the simulation study, in a multiscale change point situation (see Definition~2.1 of \cite{ChoKirch}) no single bandwidth can detect all change points. Therefore, one future topic of interest would be to extend the proposed procedure to a true multiscale setup as  in \cite{ChoKirch}.
\section*{Acknowledgements}
This work was supported by the Deutsche Forschungsgemeinschaft (DFG, German Research Foundation) - 314838170, GRK 2297 MathCoRe.

\printbibliography

\appendix
\section{Appendix: Proofs}\label{appendix}
We first prove some bounds for the limiting Wiener process that will be used throughout the proofs (for (i)) or are related to the bounds in Assumption~\ref{Hajek_Renyi} (for (ii) and (iii)).

\begin{Proposition}\label{prop_WP_rates}Let Assumption~\ref{invariance_principle} hold with a rate of convergence as in Assumption~\ref{bandwidth} with the notation of Assumption~\ref{Hajek_Renyi}. Let $0<\xi_T\le h_T$ and $D_T\ge 1$ be arbitrary sequences (bounded or unbounded).
\begin{enumerate}[(a)]
	\item		The following bounds hold for the Wiener processes as in Assumption~\ref{invariance_principle}:
		\begin{align*}
(i)\quad & 
\max_{i=1,\ldots,q_T}\sup_{0\le t\le \xi_T}\frac{1}{\sqrt{\xi_T}}\,\|\bW_{\theta_i}^{(\theta_{i+1})}-\bW_{\theta_i\pm t}^{(\theta_{i+1})}\|=O_P\left(\sqrt{\log 2q_T}\right),\\
(ii)\quad & 	\sup_{\frac{D_T}{\|\bd_i\|^2}\le s \le h_T} \frac{\sqrt{D_T} \left\|\bW_{\theta_i}^{(\theta_i)}-\bW_{\theta_i\pm s}^{(\theta_i)}\right\|}{s\left\|\bd_i\right\|}=O_P\left(1\right),\\
(iii)\quad &\max_{i=1,\ldots,q_T}\sup_{\frac{D_T}{\|\bd_i\|^2}\le s \le h_T} \frac{\sqrt{D_T} \left\|\bW_{\theta_i}^{(\theta_i)}-\bW_{\theta_i\pm s}^{(\theta_i)}\right\|}{s\left\|\bd_i\right\|}=O_P\left(\sqrt{\log 2q_T}\right).
		\end{align*}
	\item The  bound in (i) carries over to the centered increments of the original process:
		\begin{align*}
\max_{i=1,\ldots,q_T}\sup_{0\le t\le \xi_T}\frac{1}{\sqrt{\xi_T}}\,\|\bRc_{\theta_i}^{(\theta_{i+1})}-\bRc_{\theta_i\pm t}^{(\theta_{i+1})}\|=O_P\left(\sqrt{\log 2q_T}\right).
\end{align*}
The bound in (ii) carries over if a forward and backward invariance principle as above exists starting in an arbitrary point $\theta_i$, in this case (iii) carries over if $q_T=O(1)$.
	\end{enumerate}
	For a single change point (instead of taking the maximum over all) the bound in (a) (i) and (b) is given by $O_P(1)$.
\end{Proposition}
\begin{proof}
(a) Let $\bB^{(j)}_t=(\bSigma_T^{(j)})^{-1/2}\,\bW^{(j)}_t$, then  by the self-similarity of Wiener processes it holds
\begin{align*}
	& \max_{i=1,\ldots,q_T}\sup_{0\le t\le \xi_T}\frac{1}{\sqrt{\xi_T}}\,\|\bW_{c_i}^{(c_{i+1})}-\bW_{c_i+t}^{(c_{i+1})}\|\\
	&	\le O(1)\, \max_{j=1,\ldots,P}\|(\bSigma^{(j)}_T)^{1/2}\|\,  \max_{i=1,\ldots,q_T}\sup_{0\le t\le 1}\|\bB_{c_i}^{(c_{i+1})}-\bB_{c_i+t}^{(c_{i+1})}\|.
\end{align*}
By the uniform boundedness of the covariance matrices as in Assumption~\ref{invariance_principle}, $$\max_{j=1,\ldots,P}\|(\bSigma_T^{(j)})^{1/2}\|=\max_{j=1,\ldots,P}\sqrt{\|\bSigma_T^{(j)}\|}=O(1).$$ 

The reflection principle in combination with tail probabilities for Gaussian random variables shows that with  appropriate constants $D_1,D_2$ (not depending on $i$) it holds for all $D\ge 1$
\begin{align*}
	P\left( \sup_{0\le t\le 1}\|\bB_{c_i}^{(c_{i+1})}-\bB_{c_i+t}^{(c_{i+1})}\|\ge D_1 \sqrt{D\,\log(2\,q_T)} \right)\le \frac{D_2}{2^D\,q_T^D},
\end{align*}
which in combination with  subadditivity shows that 
\begin{align*}
	\max_{i=1,\ldots,q_T}\sup_{0\le t\le 1}\|\bB_{c_i}^{(c_{i+1})}-\bB_{c_i+t}^{(c_{i+1})}\|=O_P(\sqrt{\log 2q_T}).
\end{align*}
The assertion without the maximum follows analogously.

Clearly, (ii) follows from (iii) so we will only prove the latter. As above it is sufficient to prove the assertion for $\{\bB_t\}$.
Due to  the self-similarity of Wiener processes and its stationary and independent increments, it holds
\begin{align*}
	\max_{i=1,\ldots,q_T}\sup_{\frac{D_T}{\|\bd_i\|^2}\le s \le h_T} \frac{\sqrt{D_T} \left\|\bB_{c_i+s}^{(c_i)}-\bB_{c_i}^{(c_i)}\right\|}{s\left\|\bd_i\right\|}\overset{\mathcal{D}}{=}
	\max_{j=1,\ldots,q_T}\sup_{1\le t \le h_T \|\bd_j\|^2/D_T} \frac{ \left\|\bB_{t}^{(j)}\right\|}{t},
\end{align*}
where $\{\bB_t^{(j)}\}$, $j=1,2,\ldots$, are independent standard Wiener processes. Similar assertions hold for the other expressions.
By the reflection principle and tail probabilities for Gaussian random variables it holds for any $C>4$ 
\begin{align*}
	& P\left( \sup_{t\ge 1}\frac{\|\bB_t\|}{t}\ge  \sqrt{C\,\log 2q_T} \right)\le \sum_{l\ge 1}P\left( \sup_{2^l\le t<2^{l+1}} \frac{\|\bB_t\|}{t}\ge \sqrt{\,C \log 2q_T} \right)\\
	&\le \,\sum_{l\ge 1} P\left( \sup_{0\le t\le 1} \|\bB_{t}\|\ge \frac{2^l}{\sqrt{2^{l+1}}}\, \sqrt{C\,\log 2q_T} \right)
	\le\,O(1)\sum_{l\ge 1} \left(O(1)  2Cq_T \right)^{-2^l}\\
	&=O\left( \frac{1}{4C^2 q_T^2} \right),
\end{align*}
which shows the assertion in combination with the sub-additivity. 

(b) By the invariance principle and (a) (i)
it holds
\begin{align*}
	\,\max_{i=1,\ldots,q_T}\ \sup_{c_i-h_T< t\le c_i+h_T}
	\left\|\noise_t(\bRc_t)\right\|
&\le O_P\left( \frac{T^{1/2}\nu_T}{\sqrt{h_T}}  \right) +
 \max_{i=1,\ldots,q_T}\ \sup_{c_i-h_T< t\le c_i+h_T} \left\|\noise_t(\bW)\right\|\\&=O_P(\sqrt{\log 2q_T})
.
\end{align*}
The other statement can be proven analogously. For the assertion in (ii) the invariance principle starting in $\theta_i$ backward or forward applied from $ \theta_i $ to $ \theta_i\pm h_T $ yields a rate of $h_T^{1/2}\nu_{h_T}$, which is strong enough to prove the rate analogously to above.

\end{proof}

\subsection{Proofs of Section~\ref{subsection_threshold}}

\begin{proof}[Proof of Theorem \ref{theorem_threshold}]\ \\
	(a) Because $\widehat{\mathbf{A}}_t$ is symmetric and positive definite such that the minimal eigenvalue of $A_t^{-1}$ is given by $1/\|\widehat{\mathbf{A}}_t\|$ it holds
	\begin{align*}
		\signal_t\widehat{\mathbf{A}}_t^{-1}\signal_t \ge&\ \frac{1}{\|\widehat{\mathbf{A}}_t\|}\, \|\signal_t\|^2		= \frac{1}{\|\widehat{\mathbf{A}}_t\|}\,\frac{(h-\abs{c_i-t})^2}{2h} \|\bd_i\|^2.
	\end{align*}

	(b)  
		By the invariance principle from Assumption~\ref{invariance_principle} it holds by Assumption \ref{bandwidth} that
		\begin{align}\label{eq_approx_wp}
			\sup_{h\le t\le T}\left\|\noise_t-\noise_t(\mathbf{W})\right\|=O_P\left( \frac{T^{1/2}\nu_T}{\sqrt{h}} \right)=o_P\left( \sqrt{\log(T/h)}^{-1} \right),
		\end{align}
		where $\noise_{t}(\mathbf{W}_t)$ is the \mosum\  statistics defined in \eqref{eq_mosum}  with $\{\mathbf{Z}_t\}$ there replaced by $\{\mathbf{W}_t\}$. 
Assertion	 (b)(i) follows immediately by the $1/2$-self-similarity of the Wiener process with $ {\bB}_t=\bSigma_t^{-1/2}\bW_t $.\\ 
		For the sub-linear case as in (ii) we get by \eqref{eq_approx_wp}
		\begin{align*}
&	a\left(\frac{T}{h}\right)
	\supp{h\le t\le T-h}\left\|\bSigma_t^{-1/2}\noise_t\right\|
	=
	a\left(\frac{T}{h}\right)
	\supp{h\le t\le T-h}\left\|\noise_t({\bB})\right\|+o_P(1)\\
	&\overset{\mathcal{D}}{=}\frac{1}{\sqrt{2}}
\sup_{0\le s\le \frac{T}{h}-2}\left\|{\bB}_{s+2}-2{\bB}_{s+1}+{\bB}_{s}\right\|+o_P(1),
\end{align*}
where $ \left(\noise_{t}\right)_{t\ge 0} $ is a stationary process.
 Assertion (b)(ii) follows 
 by  \cite{SteinebachEastwood}, Lemma 3.1 in combination with Remark 3.1 with $\alpha=1$ and $C_1=\ldots=C_p=\frac 32$.

Replacing $ \bSigma_t $ by $ \widehat{\bSigma}_t $ does not change any of the above assertions by standard arguments.

(c) By splitting $\noise_t(\bRc)$ into increments of length at most $2h$ anchored at the change points $c_i$ we get by Proposition~\ref{prop_WP_rates}(b)(i)
\begin{align*}
	& \max_{i=1,\ldots,q_T}\ \sup_{c_i-h< t\le c_i+h} \left\|\noise_t(\bRc)\right\|\notag\\
	&=O(1)\,	\max_{i=1,\ldots,q_T}\sup_{0\le t\le h}\frac{1}{\sqrt{h}}\|\bRc_{c_i}^{(c_{i+1})}-\bRc_{c_i+t}^{(c_{i+1})}\|+O(1)\,	\max_{i=1,\ldots,q_T}\sup_{0\le t\le h}\frac{1}{\sqrt{h}} \|\bRc_{c_i}^{(c_{i})}-\bRc_{c_i-t}^{(c_{i})}\|
	\\&=O_P(\sqrt{\log 2q_T}).
\end{align*}
This shows that 
\begin{align*}
	\max_{i=1,\ldots,q_T}\ \sup_{c_i-h<t\le c_i+h} \noise_t^{\prime}\bSigma_t^{-1}\noise_t=&\ O_P(\log 2q_T),
\end{align*}
In combination with (b) and the fact that there are only finitely many regimes (c) follows.

\end{proof}

\subsection{Proofs of Section \ref{section_consistency}}

We first prove consistency of the segmentation procedure.
\begin{proof}[Proof of Theorem~\ref{theorem_consistency}] 
	Define for $0<\tau<1$ the following set
	\begin{align}\label{set_Sn}
	 &\ S_T= S_T^{(1)}\cap S_T^{(2)}\cap \bigcap_{j=1}^{q_T}\left(S_T^{(3)}(j,\tau)\cap S_T^{(4)}(j,\tau)\right),
	\end{align}
	where
	\begin{align*}
		&\ 
		S_T^{(1)}=	\geschweift{\max_{j=1,\ldots,q_T}\sup_{|t-c_j|> h}\bM_t^{\prime}\widehat{\mathbf{A}}_t^{-1}\bM_t <\thr},  
		 \\ &\ S_T^{(2)}=
		 \geschweift{\min_{j=1,\ldots,q_T}\bM_{c_j}^{\prime}\widehat{\mathbf{A}}_{c_j}^{-1}\bM_{c_j}\ge \thr}, \notag\\
		 &\ S_T^{(3)}(j,\tau)=	\bigcap_{k=1}^{\lceil\frac{1}{\tau}\rceil-1}\geschweift{\zentrierenalt{\supp{c_j-h\le t\le  c_j-k\tau h}\|\bM_t\|<\|\bM_{c_j-(k-1)\tau h}\|}},\\
	 &\ S_T^{(4)}(j,\tau)=	\bigcap_{k=1}^{\lceil\frac{1}{\tau}\rceil-1}\geschweift{\zentrierenalt{\supp{ c_j+k\tau h\le t\le c_j+h }\|\bM_t\|<\|\bM_{c_j+(k-1)\tau h\|}}}.
			\end{align*}
			On $S_T^{(1)}$ there are asymptotically no significant points outside of $h$-environments of the change points. On $S_T^{(2)}$ there is at least one significant time point for each change point. On $S_T^{(3)}(j,\tau)\cap S_T^{(4)}(j,\tau)$ with $\tau<\eta/2$, there are no local extrema (within the $h$-environment of $c_j$) that are outside the interval $(c_j-\tau h, c_j+\tau h)$. 
			Additionally, on $S_T^{(2)}\cap S_T^{(3)}(j,\tau)\cap S_T^{(4)}(j,\tau)$ the global extremum within that interval will be the only significant local extremum within the $h$-environment of $c_j$ such that 
			\begin{align*}
				\left\{\max_{i=1,\ldots,\min(\hat{q}_T,q_T)}\left|\hat{c}_i-\cporig_i\right|\le \tau h,	\hat{q}_T=q_T\right\}
\supset S_T.
			\end{align*}
We will conclude the proof by showing that $S_T$ is an asymptotic one set. 

Indeed,  $P(S_T^{(1)})\to 1$ by Theorem~\ref{theorem_threshold} (c) on noting that
\begin{align*}
\bM_t^{\prime}\widehat{\mathbf{A}}_t^{-1}\bM_t\le \|\widehat{\mathbf{A}}_t^{-1}\| \,\|\bM_t\|^2 
\end{align*}
and $P(S_T^{(2)})\to 1$ by Theorem~\ref{theorem_threshold} (a) and (c).

Similarly, for $ c_{i}-h \le t \le c_i $, we obtain that	
	\begin{align*}
		&\  \|\bM_{c_i-(k-1)\tau h}\|-\|\bM_{c_i-k\tau h}\|\ge \|\signal_{c_i-(k-1)\tau h}\|-\|\signal_{c_i-k\tau h}\|+O_P\left( \sqrt{\log(T/h)} \right)\\
		&\ge \frac{\tau}{\sqrt{2}} \sqrt{h}\,\|\mathbf{d}_i\| (1+o_P(1)), \end{align*} 		
where the $o_P$-term is uniform in $i$. This shows that $P\left( \bigcap_{j=1}^{q_T}S_T^{(3)}(j,\tau) \right)\to 1$. The assertion $P\left( \bigcap_{j=1}^{q_T}S_T^{(4)}(j,\tau) \right)\to 1$
follows analogously.
\end{proof}

With the above proposition we are ready to prove the localisation rates for the change point estimators.
	\begin{proof}[Proof of Theorem \ref{theorem_convergence_rate}]
		On $S_T$ as in \eqref{set_Sn} it holds for any sequence $\xi_T$
		\begin{align*}
			\left\{\hat{c}_i-c_i<- C\xi_T^2/\|\bd_i\|^2\right\}
			&= \left\{
				\sup_{c_i-h\le t\le c_i-C \xi^2_T/\left\|\bd_i\right\|^2}	\|\bM_t\|^2	\ge \sup_{c_i-C \xi^2_T/\left\|\bd_i\right\|^2\le t\le c_i+h}	\|\bM_t\|^2\right\}\\
				&\subset
\left\{
	\sup_{c_i-h\le t\le c_i-C \xi^2_T/\left\|\bd_i\right\|^2} 2h\,\left(	\|\bM_t\|^2-\|\bM_{c_i}\|^2\right)	\ge 0
\right\}.
			\end{align*}
			We will now show that the probability for the last set becomes arbitrarily small for $C$ sufficiently large with $\xi_T=\omega_T$ as well as that the probability for the union of these sets over all change points $i=1,\ldots,q_T$ becomes arbitrarily small for $\xi_T=\tilde{\omega}_T$. An analogous assertion can be shown for $\hat{c}_i>c_i+C\xi_T^2/\|\bd_i\|^2$, completing the proof. 
	For $c_i-h\le t<c_i$ the following decomposition holds
	\begin{align}
		&\mathbf{V}_t=\|	\bM_{t}\|^2-\|\bM_{c_i}\|^2= - (\signal_{c_i}-\signal_t+\noise_{c_i}-\noise_t  )^{\prime}\,(\signal_{c_i}+\signal_t+\noise_{c_i}+\noise_t  )\notag\\
		&=-\frac{1}{2h}\,\left(\signalsquareterms_{1,t}\,\bd_i+\diffusionsquareterms_{1,t}\right)^{\prime}\left(\signalsquareterms_{2,t}\,\bd_i+\diffusionsquareterms_{2,t}\right),\label{decomposition_statistic}\\
		&\text{where } 
		\signalsquareterms_{1,t}=c_i-t>0,\quad \signalsquareterms_{2,t}=2h+t-c_i\ge h,\notag\\[2mm]
		&\phantom{\text{where }}
\diffusionsquareterms_{1,t}=(\bRc_{c_i-h}^{(c_i)}-\bRc_{t-h}^{(c_i)})+
(\bRc_{c_i+h}^{(c_{i+1})}-\bRc_{t+h}^{(c_{i+1})})
- 2	\left(\bRc_{c_i}^{(c_i)}-\bRc_{t}^{(c_i)}\right)\notag\\
&\phantom{\text{where }}
\diffusionsquareterms_{2,t}=(\bRc_{c_i+h}^{(c_{i+1})}-\bRc_{t+h}^{(c_{i+1})})
		+ 2
		\left(\bRc_{t+h}^{(c_{i+1})}-\bRc_{c_i}^{(c_{i+1})}\right)
	 -(\bRc_{c_i-h}^{(c_i)}-\bRc_{t-h}^{(c_i)})\notag\\*
	 &\phantom{\text{where }\diffusionsquareterms_{2,t}=}	-2
		\left(\bRc_{t}^{(c_i)}-\bRc_{c_i-h}^{(c_i)}\right).\notag
	\end{align}

	We will concentrate on the proof of (b), where the proof of (a) is done analogously without the maximum over the change points and with the (possibly) tighter rate $\omega_T$ as in Assumption~\ref{Hajek_Renyi} (a) instead of $\tilde{\omega}_T$ as in (b). Indeed, Assumption~\ref{Hajek_Renyi} (b) immediately implies that for any $\epsilon>0$ there exists a $C$ such that for any $y>0$ it holds
	\begin{align*}
		&P\left( \max_{i=1,\ldots,q_T}\sup_{c_i-h\le t\le c_i-C \tilde{\omega}_T^2/\|\bd_i\|^2}\frac{\|\diffusionsquareterms_{1,t}\|}{\signalsquareterms_{1,t}\,\|\bd_i\|}\ge y \right)
		\\&=
		P\left( \max_{i=1,\ldots,q_T}\sup_{C \tilde{\omega}_T^2/\|\bd_i\|^2\le c_i-t\le h }\sqrt{C\tilde{\omega}_T^2}\,\frac{ \|\diffusionsquareterms_{1,t}\|}{\|\bd_i\||c_i-t|}\ge \sqrt{C}\, y\, \tilde{\omega}_T \right)
		\le \epsilon.\notag
	\end{align*}
	Similarly, by Proposition~\ref{prop_WP_rates} (b)(i), it holds
	\begin{align*}
		&	 \max_{i=1,\ldots,q_T}\sup_{c_i-h\le t\le c_i-C \tilde{\omega}_T^2/\|\bd_i\|^2}	\frac{\|\diffusionsquareterms_{2,t}\|}{\signalsquareterms_{2,t}\,\|\bd_i\|}=O_P\left( \sqrt{\frac{\log 2q_T}{h \|\bd_i\|^2}} \right)=o_P(1),
	\end{align*}
	where the last statement follows by Assumption~\ref{bandwidth} on noting that $q_T\le T/(2h)$.

	Combining the above assertions with $P(S_T^c)=o_P(1)$  we obtain using the Cauchy-Schwarz inequality
	\begin{align*}
		&P\left(\|\bd_i\|^2(\hat{c}_i-c_i)<- C\tilde{\omega}_T^2 \text{ for some }i=1,\ldots,q_T \right)\\
		&\le o_P(1)+ P\Biggl(\max_{i=1,\ldots,q_T}\supp{c_i-h\le t\le c_i-\frac{C\tilde{\omega}_T^2}{\left\|\bd_i\right\|^2}}-D_{1,t}D_{2,t}\|\bd_i\|^2\,\\*
&\phantom{{\max_{i=1,\ldots,q_T}\supp{c_i-h\le t\le c_i-\frac{C\tilde{\omega}_T^2}{\left\|\bd_i\right\|^2}}-D_{1,t}
}}
		\cdot \left(1+\frac{\diffusionsquareterms_{1,t}^{\prime}\bd_i}{\signalsquareterms_{1,t}\|\bd_i\|^2}+\frac{\bd_i^{\prime}\diffusionsquareterms_{2,t}}{\signalsquareterms_{2,t}\|\bd_i\|^2}+\frac{\diffusionsquareterms_{1,t}^{\prime}\diffusionsquareterms_{2,t}}{\signalsquareterms_{1,t}\signalsquareterms_{2,t}\,\|\bd_i\|^2}\right)\ge 0 \Biggr)\\
	& \le o_P(1)+ \Prob{\max_{i=1,\ldots,q_T}\supp{c_i-h\le t\le c_i-\frac{C\tilde{\omega}_T^2}{\left\|\bd_i\right\|^2}}\left|\frac{\diffusionsquareterms_{1,t}^{\prime}\bd_i}{\signalsquareterms_{1,t}\|\bd_i\|^2}+\frac{\bd_i^{\prime}\diffusionsquareterms_{2,t}}{\signalsquareterms_{2,t}\|\bd_i\|^2}+\frac{\diffusionsquareterms_{1,t}^{\prime}\diffusionsquareterms_{2,t}}{\signalsquareterms_{1,t}\signalsquareterms_{2,t}\,\|\bd_i\|^2}\right|\ge 1}\\
	&\le o_P(1) + P\left(\max_{i=1,\ldots,q_T}\supp{c_i-h\le t\le c_i-\frac{C\tilde{\omega}_T^2}{\left\|\bd_i\right\|^2}} \,\frac{\|\diffusionsquareterms_{1,t}\|}{\signalsquareterms_{1,t}\|\bd_i\|}\ge \frac 1 3  \right)\le \epsilon
	\end{align*}
	for $C$ large enough (and $\epsilon$ arbitrary). This concludes the proof.
	\end{proof}
		
\begin{proof}[\textbf{Proof of Theorem \ref{limit_distribution}}]\ 
	For $ 0\le c_i-t\le D/\|\bd_i\|^2$ it holds by Proposition~\ref{prop_WP_rates}	(b) (i) (the result without the maximum over all change points) with the notation as in \eqref{decomposition_statistic}
	\begin{align*}
		&\max_{0\le c_i-t\le D/\|\bd_i\|^2}\|\diffusionsquareterms_{1,t} \| =O_P\left( \frac{1}{\|\bd_i\|} \right),\qquad \max_{0\le c_i-t\le D/\|\bd_i\|^2}\|\diffusionsquareterms_{2,t} \| =O_P\left( \sqrt{h} \right),\\
		&\max_{0\le c_i-t\le D/\|\bd_i\|^2}|\signalsquareterms_{1,t} | =O\left( \frac{1}{\|\bd_i\|^2} \right), \quad \max_{0\le  c_i-t\le D/\|\bd_i\|^2}|\signalsquareterms_{2,t}-2h | =O\left( \frac{1}{\|\bd_i\|^2} \right).
	\end{align*}
	Together with \eqref{decomposition_statistic} this shows
	\begin{align*}
		&		\mathbf{V}_t=-\|\bd_i\|^2 |c_i-t|- \|\bd_i\|\, \diffusionsquareterms_{1,t}^{\prime}\mathbf{u}_i+O_P\left( \frac{1}{\sqrt{h} \,\|\bd_i\|} \right)	\end{align*}
	By Assumption~\ref{bandwidth} it holds $\|\bd_i\|^2h\to \infty$ such that with the substitution \mbox{$s=(t-c_i)\|\bd_i\|^2$} with $-D\le s\le 0$ we get
	\begin{align*}
		&\mathbf{V}_s=-|s|+\|\bd_i\|\,\left( \bY_{D+s}^{(1)}-\bY_D^{(1)} \right)^{\prime}\bu_i-2 \|\bd_i\|\,\left(\bY_{D+s}^{(21)}-\bY_D^{(21)} \right)^{\prime}\bu_i\\*
		&\phantom{\mathbf{V}_s=-|s|}
			+\|\bd_i\|\,\left( \bY_{D+s}^{(3)}-\bY_D^{(3)} \right)^{\prime}\bu_i+o_P(1).
	\end{align*}
	Similarly, for $0\le t-c_i\le D/\|\bd_i\|^2$ and the same substitution now leading to $0\le s\le D$ we get
	\begin{align*}
		&\mathbf{V}_s=-|s|+\|\bd_i\|\,\left( \bY_{D+s}^{(1)}-\bY_D^{(1)} \right)^{\prime}\bu_i-2 \|\bd_i\|\,\left(\bY_{D+s}^{(22)}-\bY_D^{(22)} \right)^{\prime}\bu_i\\*
		&\phantom{\mathbf{V}_s=-|s|}
			+\|\bd_i\|\,\left( \bY_{D+s}^{(3)}-\bY_D^{(3)} \right)^{\prime}\bu_i+o_P(1).
	\end{align*}
	Note that for $\|\bd_i\|^2\,|\hat{c}_i-c_i|\le D$ it holds
	\begin{align*}
		\|\bd_i\|^2 (\hat{c}_i-c_i)\le x \quad\iff\quad \max_{-D\le s\le x}\mathbf{V}_s\ge \max_{x<s\le D}\mathbf{V}_s.
	\end{align*}
	Now, first applying the functional central limit theorem from Assumption~\ref{ass_cp_distr} and then letting $D\to \infty$ (in combination with Theorem~\ref{theorem_convergence_rate}, where now by assumption $\omega_T=1$) yields the result.
\end{proof}

		\end{document}